\documentclass[11pt,a4paper,reqno]{amsart}
\usepackage{amsmath}
\usepackage[english]{babel}
\usepackage{latexsym}
\usepackage{amssymb}
\usepackage{amscd}
\usepackage{amsgen,amstext,amsbsy,amsopn,amsfonts}
\usepackage{bm,bbm}
\usepackage{amsthm,epsfig,graphicx,graphics}
\usepackage[utf8]{inputenc}
\usepackage{xspace}
\usepackage{amsxtra}
\usepackage{color}
\usepackage{dsfont}
\usepackage{enumerate}
\usepackage{tikz}
\usepackage{hyperref}
\usetikzlibrary{angles}
\usetikzlibrary{quotes}
\usetikzlibrary{patterns}
\usepackage[font=footnotesize]{caption}
\usepackage[sort&compress,capitalise,nameinlink]{cleveref}
\crefname{section}{\textsection}{\textsection}
\crefname{subsection}{\textsection}{\textsection}
\crefname{subsubsection}{\textsection}{\textsection}
\usepackage{mathrsfs}

\usepackage[hmargin={2.7cm,2.7cm},vmargin={2.5cm,2.5cm},centering]{geometry}

\DeclareMathAlphabet{\mathpzc}{OT1}{pzc}{m}{it}

\makeatletter
\renewcommand*{\@textcolor}[3]{%
  \protect\leavevmode
  \begingroup
    \color#1{#2}#3%
  \endgroup
}
\makeatother

\numberwithin{equation}{section}
\newcommand{\bdm}{\begin{displaymath}}
\newcommand{\edm}{\end{displaymath}}
\newcommand{\bay}{\begin{array}{c}}
\newcommand{\eay}{\end{array}}
\newcommand{\ben}{\begin{enumerate}}
\newcommand{\een}{\end{enumerate}}
\newcommand{\beq}{\begin{equation}}
\newcommand{\eeq}{\end{equation}}
\newcommand{\beqn}{\begin{eqnarray}}
\newcommand{\eeqn}{\end{eqnarray}}
\newcommand{\bml}[1]{\begin{multline} #1 \end{multline}}
\newcommand{\bmln}[1]{\begin{multline*} #1 \end{multline*}}

\newcommand{\lf}{\left}
\newcommand{\ri}{\right}

\newcommand{\xv}{\mathbf{x}}

\newcommand{\rv}{\mathbf{r}}

\newcommand{\nv}{\mathbf{n}}
\newcommand{\fv}{\mathbf{F}}
\newcommand{\ev}{\mathbf{e}}
\newcommand{\gamv}{\bm{\gamma}}

\newcommand{\diff}{\mathrm{d}}
\newcommand{\eps}{\varepsilon}

\newcommand{\dist}{\mathrm{dist}}

\newcommand{\as}{\alpha_{\star}}

\newcommand{\nuv}{\bm{\nu}}
\newcommand{\tav}{\bm{\tau}}

\renewcommand{\ss}{\mathsf{s}}
\renewcommand{\tt}{\mathsf{t}}

\newcommand{\glm}{\psi^{\mathrm{GL}}}

\newcommand{\aav}{\mathbf{A}}

\newcommand{\theo}{\Theta_0}

\newcommand{\dom}{\mathscr{D}}

\newcommand{\eones}{E^{\mathrm{1D}}_{\star}}

\newcommand{\fol}{f_{0}}
\newcommand{\Fol}{F_{0}}
\newcommand{\kol}{K_{0,\ell}}
\newcommand{\fs}{f_{\star}}

\newcommand{\onedom}{\mathscr{D}^{\mathrm{1D}}}

\newcommand{\curl}{\mathrm{curl}}

\newcommand{\ann}{\mathcal{A}}
\newcommand{\anne}{\mathcal{A}_{\eps}}

\newcommand{\trial}{\psi_{\mathrm{trial}}}

\newcommand{\disp}{\displaystyle}
\newcommand{\tx}{\textstyle}

\newcommand{\R}{\mathbb{R}}

\newcommand{\E}{\mathcal{E}}

\newcommand{\OO}{\mathcal{O}}

\newcommand{\al}{\alpha}

\newcommand{\Om}{\Omega}

\newcommand{\one}{\mathds{1}}

\newcommand{\Hcc}{H_{\mathrm{c}2}}
\newcommand{\Hccc}{H_{\mathrm{c}3}}
\newcommand{\Hstar}{H_{\mathrm{corner}}}

\newtheorem{teo}{Theorem}[section]
\newtheorem{lem}{Lemma}[section]
\newtheorem{pro}{Proposition}[section]
\newtheorem{cor}{Corollary}[section]

\newtheorem{asum}{Assumption}
\newtheorem{conj}{Conjecture}

\theoremstyle{remark}
\newtheorem{remark}{Remark}[section]

\newcommand{\annol}{I_{\bar{\ell}}}

\newcommand{\fone}{\E^{\mathrm{1D}}}

\newcommand{\eone}{E^{\mathrm{1D}}}

\newcommand{\eoneo}{E ^{\rm 1D}_{0}}

\newcommand{\alO}{\alpha_0}

\newcommand{\jv}{\mathbf{j}}

\renewcommand{\leq}{\leqslant}
\renewcommand{\geq}{\geqslant}
\renewcommand{\Im}{\mathrm{Im}}

\newcommand{\corner}{\Gamma_{\beta}(L,\ell)}

\newcommand{\bdo}{\partial \Gamma_{\mathrm{out}}}
\newcommand{\bdi}{\partial \Gamma_{\mathrm{in}}}
\newcommand{\bdbd}{\partial \Gamma_{\mathrm{bd}}}

\newcommand{\gpm}{\Gamma^{\pm}}
\newcommand{\gcomp}{\Gamma_{\mathrm{c}}}
\newcommand{\ggam}{\Gamma_{\gamma}}
\newcommand{\gcompm}{\Gamma_{\mathrm{c}}^{\pm}}
\newcommand{\ggampm}{\Gamma_{\gamma}^{\pm}}
\newcommand{\thebis}{\vartheta_{\mathrm{bis}}}

\newcommand{\ecorr}{E_{\mathrm{corr}}}
\newcommand{\ecorn}{E_{\mathrm{corner},\beta}}

\newcommand{\exl}{\OO(\ell^{-\infty})}

\begin{document}

\title{Almost Flat Angles in Surface Superconductivity}

\author[M. Correggi]{Michele Correggi}
\address{Dipartimento di Matematica, Politecnico di Milano, P.zza Leonardo da Vinci, 32, 20133, Milano, Italy.}
\email{michele.correggi@gmail.com}
\urladdr{https://sites.google.com/view/michele-correggi}

\author[E.L. Giacomelli]{Emanuela L. Giacomelli}
\address{Department of Mathematics, Ludwig-Maximilians Universit\"at M\"unchen, Theresienstr., 39, 80333, M\"unchen, Germany.}
\email{emanuela.giacomelli.elg@gmail.com}

\date{\today}

\begin{abstract}
	Type-II superconductivity is known to persist close to the sample surface in presence of a strong magnetic field. As a consequence, the ground state energy in the Ginzburg-Landau theory is approximated by an effective one-dimensional model. As shown in \cite{CG2}, the presence of corners on the surface affects the energy of the sample with a non-trivial contribution. In \cite{CG2}, the two-dimensional model problem providing the corner energy is implicitly identified and, although no explicit dependence of the energy on the corner opening angle is derived, a conjecture about its form is proposed. We study here such a conjecture and confirm it, at least to leading order, for corners with almost flat opening angle.
\end{abstract}

\maketitle

\tableofcontents

\section{Introduction}\label{sec:intro}
The phenomenon of {\it superconductivity} was first discovered in 1911 by H.\ Kamerlingh Onnes and it is nowadays well understood: the electrical resistance of various materials, such as mercury, drops down dramatically below a critical temperature $T_c$, whose value depends on the material. Its microscopic explanation relies on BCS theory \cite{BCS}, which describes superconductivity as a quantum critical phenomenon in which the conducting electrons arrange in weakly bound pairs (Cooper pairs) and, once they are created, the pairs exhibit a collective behavior, which is similar to the one appearing in Bose-Einstein condensate and which is responsible of the sudden drop of resistivity.

When a superconductor is immersed in a magnetic field, however, an even richer physics may emerge: for instance, if the intensity of the external field is small enough, the material behaves like a perfect superconductor. Hence, no effect of the field is observed in the sample and the field is expelled from the material, giving rise to the famous {\it Meissner effect}. However, if the field is strong enough, the material loses all its superconducting properties. According to how this breakdown  occurs, one can distinguish between two different kinds of superconductors: for type-I superconductors, there is only a first order phase transition, while type-II superconductors exhibit a more complex behavior and the normal and superconducting phase may coexist (mixed state).

Close enough to the critical temperature, the response of the superconductor to the external magnetic field is very well described through a macroscopic model proposed in 1950 by V.L. Ginzburg and L.D. Landau \cite{GL} -- the Ginzburg-Landau (GL) theory --, where all the information is encoded in a one-particle wave function known as {\it order parameter} together with the magnetic potential generating the observed magnetic field. Such a theory is clearly much simpler than the microscopic description provided by BCS theory, but nevertheless its predictions are extremely accurate: the only external parameters are indeed the intensity and direction of the applied field and a length ({\it London penetration depth}), which is characteristic of each material and appears in the theory through a parameter $ \kappa > 0 $, whose threshold value $ \kappa = 1 $ conventionally separates type-I from type-II superconductors.

\subsection {Ginzburg-Landau Theory} The GL free energy of a type-II superconducting sample made of an infinite wire of constant cross section $\Om$ is given by
\beq\label{eq: GL funct k}
	\mathcal{G}^{\mathrm{GL}}_\kappa[\psi,\aav] = \int_\Om\, \diff\rv\, \lf\{|(\nabla + i h_{\mathrm{ex}}\aav)\psi|^2 - \kappa^2|\psi|^2 + \tx\frac{1}{2}\kappa^2|\psi|^4\ri\} + h_{\mathrm{ex}}^2\int_{\mathbb{R}^2}\,\diff\rv\, |\mathrm{curl}\aav-1|^2.
\eeq
where we assumed that the applied magnetic field is of uniform intensity $h_{\mathrm{ex}}$ along the superconducting wire and that it is perpendicular to $\Om$; $h_{\mathrm{ex}}\aav: \mathbb{R}^2\rightarrow \mathbb{R}^2$ is the vector potential generating the induced magnetic field $h_{\mathrm{ex}}\mathrm{curl}\aav = h_{\mathrm{ex}}(\partial_1 \aav_2 - \partial_2\aav_1)$; $\psi: \Om\rightarrow \mathbb{C}$ is the order parameter, i.e., the center of mass wave function of the Cooper pairs. The modulus $|\psi|$ is a measure of the relative density of the superconducting Cooper pairs, i.e., $0\leq |\psi|\leq 1$ and, wherever $|\psi| = 0$, there are no Cooper pairs (loss of superconductivity), whereas, if $|\psi| = 1$ somewhere, then all the electrons are superconducting. The phase of $\psi$ encodes the information about the stationary current flowing in the superconductor, i.e.,
\beq\label{eq: sup current}
	\textbf{j}[\psi]:= \tx\frac{i}{2}(\psi\nabla\psi^\ast - \psi^\ast\nabla\psi ) = \Im(\psi^\ast\nabla\psi).
\eeq
We are interested in studying the equilibrium state of the sample, which is obtained by minimizing \eqref{eq: GL funct k} w.r.t. $\psi$ and $\aav$. Moreover, we focus on (extreme) type-II superconductors: indeed, we study the limit $\kappa\rightarrow \infty$, also known as London limit.

\subsection{Sample with corners} Before proceeding further, we specify the assumption we make on $\Om$: we consider a bounded and simply connected domain $\Om$ with piecewise smooth boundary, such that the unit inward normal $\nuv$ to the boundary is well defined everywhere but in a {\it finite} number of points. We call these points the \textit{corners} of $\Om$. We now state our assumption in a more precise way (see also \cite{Gr} for a detailed discussion of domains with non-smooth boundaries).

\begin{asum}[Piecewise smooth boundary]\label{asum: 1}
	\mbox{}	\\
	Let $\Om$ be a bounded and simply connected open set of $\mathbb{R}^2$. We assume that $\partial \Om$ is a smooth curvilinear polygon, i.e., for every $\rv\in \partial\Om$ there exists a neighborhood $U$ of $\rv$ and a map $\Phi: U\rightarrow \mathbb{R}^2$, such that
\begin{itemize}
	\item  $\Phi$ is injective;
	\item  $\Phi$ together with $\Phi^{-1}$ (defined from $\Psi(U))$ are smooth;
	\item  the region $\Om\cap U$ coincides with either $\{\rv\in \Omega\cap U\, |\, (\Phi(r))_1<0\}$ or $\{\rv\in \Om\cap U\, |\, (\Phi(\rv))_2 <0\}$ or $\{\rv\in \Om \cap U\, |\, (\Phi(\rv))_1 <0, (\Phi(\rv))_2<0\}$, where $(\Phi)_j$ stands for the j-th component of $\Phi$.
\end{itemize}
\end{asum}

\begin{asum}[Boundary with corners]\label{asum: 2}
	\mbox{}	\\
We assume that the set $\Sigma$ of corners of $\partial\Om$ is non-empty but finite and we denote by $\beta_j$ the angle of the $j-th$ corner (measured towards the interior).
\end{asum}

\subsection{Critical fields} We now investigate the response of the superconductor to the external magnetic field. It is well known that the material exhibits a different behavior depending on the intensity of the field. Let us first describe what occurs for samples with {\it smooth} cross sections $\Om$: one can identify three different critical values of the applied field marking phase transitions of the material. More precisely, the first critical value $H_{c1}$ is such that the minimizing order parameter has at least one vortex as soon as $h_{\mathrm{ex}}>H_{c1}$, i.e., superconductivity is lost at isolated points; the second critical field $H_{c2}$ is related to the transition form bulk to boundary behavior, meaning that if $h_{\mathrm{ex}}> H_{c2}$, superconductivity survives only near the boundary of the sample; the last critical field $H_{c3}$ is such that above it the material behaves as a normal conductor. Asymptotically, as $ \kappa \to + \infty $,
\beq\label{eq: critical fields}
	H_{c1} \sim C_\Om\log\kappa, \qquad H_{c2}\sim \kappa^2,\qquad H_{c3}\sim \frac{1}{\Theta_0}\kappa^2 
\eeq
where $C_\Om > 0 $ depends only on the domain $\Om$ and $\Theta_0\simeq 0.59$ is a universal constant. For more details about the critical fields for smooth domains, we refer respectively to \cite[\textsection 2]{SS} and \cite[\textsection\textsection 10.6 \& 13]{FH1}.

Let us now underline the expected differences in presence of corners. The first one is a possible change of the asymptotic value of the third critical field \cite{BNF}. More precisely, one observes the transition to the normal state only for fields whose intensity is larger than
\beq\label{eq: 3rd critical field corner}
 H_{c3} = \frac{1}{\mu(\beta)}\kappa^2,
\eeq
where $\mu(\beta)$ stands for the ground state energy of a Schr\"odinger operator with uniform magnetic field in the infinite sector $W_\beta$ with opening angle $\beta$, i.e.,
\[
	\mu(\beta) := \inf\mathrm{spec}_{L^2(W_\beta)} \lf(- \lf(\nabla + \tx\frac{1}{2}i\mathrm{x}^\perp \ri)^2 \ri)
\]
Note that the universal constant $\Theta_0$ related to the third critical field for smooth domains (see \eqref{eq: critical fields}) is nothing but the ground state energy of the same Schr\"odinger operator on the half-plane (i.e., for $ \beta = \pi $).

The shift of the third critical field occurs then whenever there is a corner with angle $ \beta $ such that $ \mu(\beta) < \theo $. This is proven for $ \beta \leq \frac{\pi}{2} + \epsilon$ \cite{Bo,Ja,ELP} but, based on numerical experiments, conjectured to be true \cite{Bo, ELP} for any {\it acute} angle $ 0 < \beta < \pi $. Moreover, as the applied field gets closer to \eqref{eq: 3rd critical field corner} from below, the order parameter concentrates around the corner with smallest opening angle and decays exponentially on a scale $ \kappa^{-1} $ far from it \cite[Thm. 1.6]{BNF}. Hence, before disappearing, superconductivity survives only close to the corner(s). In \cite{CG,CG2}, we proved that, if 
\beq
	\kappa^2< h_{\mathrm{ex}}< \frac{1}{\Theta_0} \kappa^2,
\eeq
 superconductivity is however uniformly distributed (in $L^2$ sense) along the boundary and exponentially small in the bulk, so recovering the same behavior as in smooth domains. This strongly suggests the emergence of a new critical field 
\beq
	\Hstar \sim \frac{1}{\Theta_0} \kappa^2
\eeq
marking the transition from surface to corner superconductivity.


\bigskip

\begin{small}
\noindent\textbf{Acknowledgments.} The authors are thankful to \textsc{S. Fournais} and \textsc{N. Rougerie} for useful comments and remarks about this work. The support of the National Group of Mathematical Physics (GNFM) of INdAM through Progetto Giovani 2016 ``Superfluidity and Superconductivity'' and Progetto Giovani 2018 ``Two-dimensional Phases'' is also acknowledged. 
\end{small}

\section{Main Results}
\label{sec: main result}

As anticipated, we aim at studying a model sample in the surface superconductivity regime, which is identified by an intensity of the external field 
 \beq\label{eq: surf sup regime}
 		h_{\mathrm{ex}} := b\kappa^2, \qquad 1< b <\Theta_0^{-1},
 \eeq
since, in this parameter windows, superconductivity survives only near the boundary of the sample. More precisely, one can prove that, if \eqref{eq: surf sup regime} holds, the minimizing order parameter exponentially decays in the distance from the boundary $\partial \Om$. We stress that the presence of corners does not affect this behavior \cite[\textsection 15.3.1]{FH1}. 

\subsection{Surface superconductivity in domains with corners}
Before discussing further the properties of any minimizing configuration, we perform a change of units, which proves to be very convenient in the surface regime given by \eqref{eq: surf sup regime} above. We introduce a new parameter
\begin{equation}
	\eps := b^{-\frac{1}{2}}\kappa^{-1}\ll 1,
\end{equation}
so that the GL functional becomes
\begin{equation}\label{eq: GL functional eps}
	\mathcal{E}^{\mathrm{GL}}_\eps[\psi, \aav;\Omega] := \int_\Omega\, \diff\rv\, \left\{\left|\left(\nabla + i \tx\frac{\aav}{\eps^2}\right)\psi\right|^2 - \disp\frac{1}{2b\eps}(2|\psi|^2 - |\psi|^4)\right\} + \frac{1}{\eps^4}\int_{\mathbb{R}^2}\, \diff\rv\, |\curl\aav - 1|^2, 
\end{equation}
We then set $E^{\mathrm{GL}}_\eps := \min_{(\psi, \aav)\in \dom^{\mathrm{GL}}}\mathcal{E}^{\mathrm{GL}}_\eps [\psi, \aav; \Omega]$,
where the minimization domain is given by 
\begin{equation}
	\dom^{\mathrm{GL}} := \lf\{(\psi, \aav)\in H^1(\Omega)\times H^1_{\mathrm{loc}}(\mathbb{R}^2; \mathbb{R}^2)\, \big|\, \mathrm{curl}\aav - 1\in L^2(\mathbb{R}^2) \ri\},
\end{equation}
and we denote by $(\psi^{\mathrm{GL}}, \aav^{\mathrm{GL}})$ any corresponding minimizing configuration.

The precise statement of the order parameter decay in this setting takes the form of the so-called Agmon estimates:

 \begin{teo}[Agmon estimates \mbox{\cite[\textsection 15.3.1]{FH1}}]
 	\label{teo: agmon}
 	\mbox{}	\\
 Let $\Om\subset\mathbb{R}^2$ be a bounded and simply connected domain, if $b > 1$. Then, any critical point $(\psi, \aav)$ of the GL functional satisfies\footnote{We denote by $ C $ a positive finite constant, whose value may change from line to line.}
 	\begin{equation}
		\int_{\Om}\diff\rv\, e^{\frac{c(b) \dist(\rv,\Om)}{\eps}}\bigg\{|\psi|^2 + \eps^2\bigg|\bigg(\nabla + i\frac{\aav}{\eps^2}\bigg)\psi\bigg|^2\bigg\}\leq C \int_{\{\dist(\rv, \partial\Om)\leq \eps\}}\diff\rv\, |\psi|^2
 	\end{equation}
 	for some $ c(b) > 0 $.
 \end{teo}
 
Thanks to the above result, one can restrict the analysis to a boundary layer of thickness of order $ \eps $, since the energy contribution of the rest of the sample is small: for instance, if we restrict to the region
\beq
	\anne : = \lf\{ \rv \in \Omega \: \big| \: \dist(\rv,\partial \Omega) \leq c_0 \eps |\log\eps| \ri\},
\eeq
for a (arbitrarily large) constant $ c_0 > 0 $, the above \cref{teo: agmon} guarantees that the energy in $ \Omega \setminus \anne $ is smaller than a (arbitrarily large) power of $ \eps $, which we denote by saying that it is $ \OO(\eps^{\infty}) $ there. The restriction to $ \ann $ is quite relevant, since it allows to use suitable tubular coordinates: let $ \gamv: [0,|\partial \Omega|) \to \partial \Omega $ be a parametrization of the boundary, then we denote by $(\ss,\tt) \in [0, |\partial \Omega|) \times [0, c_0 \eps |\log\eps|] $ the coordinates satisfying
\begin{equation}\label{eq: diff}
	\mathbf{r}(\ss,\tt) =: \gamv^\prime(\ss) + \tt \nuv(\ss), \qquad \forall \rv\in \anne,
\end{equation}
where $ \nuv $ stands for the inward normal to the boundary and $ \tt := \dist (\rv, \partial\Omega)$. The relation \eqref{eq: diff} introduce a local diffeomorphism in the smooth portion of $ \anne $.  We also denote by $(s,t)$ the $\eps-$rescaled counterparts of $ (\ss, \tt) $, i.e., $s:=\ss/\eps$, $t:= \tt\eps$. The curvature of the boundary is denoted by $\mathfrak{K}(\ss)$ and we set $ k(s) : = \mathfrak{K}(\eps s) $ for short, so that, e.g., $ \diff\rv = \eps^2 \diff t\, \diff s\, (1-\eps k(s))t $.

We now focus on the energy asymptotics in the surface superconductivity regime. The most accurate result about is proven in \cite{CG2} (see also the review \cite{Cor}) and reads as follows.

\begin{teo}[GL energy in piecewise smooth domains \mbox{\cite[Thm. 2.1]{CG2}}]
	\label{thm: CG1}
	\mbox{}		\\
	Let $\Omega\subset\mathbb{R}^2$ be any bounded simply connected domain satisfying \cref{asum: 1} and \cref{asum: 2}. Then, for any fixed $1 < b < \Theta_0^{-1}$, as $\eps\rightarrow 0$, it holds, 
	\begin{equation}\label{eq: asym CG2}
		E^{\mathrm{GL}}_\eps = \frac{|\partial\Omega| E^{\mathrm{1D}}_\star}{\eps} - E_\mathrm{corr}\int_0^{|\partial\Om|}\, \diff \ss \, \mathfrak{K}(\ss) + \sum_{j=1}^N E_{\mathrm{corner}, \beta_j} + o(1).
	\end{equation}
\end{teo}

The first three terms on the r.h.s. \eqref{eq: asym CG2} can be interpreted as follows:
	\begin{itemize}
		\item the leading order term of order $ \eps^{-1} $ is proportional to the length of the boundary and its coefficient is given by a one-dimensional model problem (see \eqref{eq: def E1D} below), whose expression is independent of the boundary curvature;
		\item the curvature corrections to the energy are contained in the second term, which is of order $ 1 $ and whose coefficient is also curvature-independent (see \eqref{eq: def Ecorr});
		\item the contribution of corners is also of order $ 1 $ and it is given by the third term in the energy expansion in terms of an effective model on a wedge-like region (see \eqref{eq: def Ecorner}).
	\end{itemize}
Both the quantities $E^{\mathrm{1D}}_\star$ and $E_{\mathrm{corr}}$ are determined through an effective one-dimensional model describing the variation of the modulus of the order parameter along the normal to the boundary. More precisely, we set
\begin{equation}\label{eq: def E1D}
	E^{\mathrm{1D}}_\star := \inf_{\alpha\in\mathbb{R}}\inf_{f\in\dom^{\mathrm{1D}}} \fone_{\alpha}[f; \R^+],	
\eeq
\beq
	\label{eq: def fone}
	\fone_{\alpha}[f; \R^+] : = \int_0^{+\infty}
\, \diff t\, \left\{|\partial_t f|^2 + (t+\alpha)^2f^2 -\frac{1}{2b}(2f^2 - f^4)\right\},
\end{equation}
with $ \dom^{\mathrm{1D}} : = \lf\{ f \in H^1(\R) \: | \: t f \in L^2(\R) \ri\} $, and
\begin{equation}\label{eq: def Ecorr}
	E_{\mathrm{corr}} := \int_0^{+\infty}\, \diff  t \: t \left\{|\partial_t f_\star|^2 + \fs^2\left(-\alpha_\star (t+\alpha_\star) -\frac{1}{b} + \frac{1}{2b}f_\star^2\right)\right\} = \frac{1}{3}f_{\star}^2(0)\alpha_\star - E^{1D}_\star,
\end{equation}
where $\alpha_\star,f_\star$ is a minimizing pair realizing $E^{\mathrm{1D}}_\star$, whose existence in discussed in \cite[\textsection A.1]{CG2}. Here, we skip most of the details about the one-dimensional models and refer to \cref{sec: appendix} or \cite{CR1,CR2,CG,CG2} instead (see in particular \cite[\textsection A]{CR2}).

\subsection{Corner effective energy}	
The corners' energy contribution is given in terms of the effective energy
\begin{equation}\label{eq: def Ecorner}
	E_{\mathrm{corner, \beta}}:= \lim_{\ell \rightarrow + \infty}\lim_{L\rightarrow + \infty}\left(E_{\mathrm{corner, \beta}}(L,\ell)\right),
\end{equation}
with
\begin{equation}\label{eq: Ecorner}
	E_{\mathrm{corner, \beta}}(L,\ell):=-2LE^{\mathrm{1D}}_0(\ell) + \inf_{\psi\in\dom_\star(\Gamma_\beta(L,\ell))}\mathcal{E}^{\mathrm{GL}}_1[\psi, \fv; \Gamma_\beta(L,\ell)]
\end{equation}
where $ \fv : = \frac{1}{2} \xv^{\perp} $, the latter functional is defined as in \eqref{eq: GL functional eps} with $ \eps  = 1 $ and $ \aav = \fv $. Moreover, $E^{\mathrm{1D}}_0(\ell)$ is the ground state energy of the finite-interval version of \eqref{eq: def E1D}, i.e.,
\begin{equation}\label{eq: def E1Dell}
	E^{\mathrm{1D}}_0(\ell) := \inf_{\alpha\in\mathbb{R}}\inf_{f\in H^1([0,\ell])}\fone_{\alpha}[f; [0,\ell]].
\end{equation}
The domain $ \Gamma_\beta(L,\ell) $ is on the other hand a sort of wedge (depicted in \cref{fig: 1}), whose longitudinal and tangential length equals $ L $ and $ \ell $, respectively, and whose opening angle is $ \beta \in (0, 2\pi) $. Note that in order for the construction to be possible, we have to add the condition $\ell\leq \tan(\beta/2)L$. Such a domain is indeed meant as a blow up of the corner region, where the longitudinal boundaries of $ \Gamma_\beta(L,\ell) $ are the $ \eps$-rescaled  (and straighten) portions of $ \partial \Omega $ close to the corner, which is represented by the vertex in $ \Gamma_\beta(L,\ell) $. The other boundaries of $ \Gamma_\beta(L,\ell) $ are in fact fake boundaries separating the corner region from the rest of the boundary layer. 
Finally, the minimization domain for the corner problem is
\begin{equation}
	\label{eq: dom star}
	\dom_\star(\Gamma_\beta(L,\ell)) :=\big\{\psi \in H^{1}(\Gamma_\beta(L,\ell)), \vert\, \psi\vert_{\partial\Gamma_{\mathrm{bd}}\cup\partial\Gamma_{\mathrm{in}}} = \psi_\star\big\},
\end{equation}
where we denoted by $\partial\Gamma_{\mathrm{bd}}\cup \partial\Gamma_{\mathrm{in}}$ the inner and tangential boundaries of $ \Gamma_{\beta} $, i.e., concretely, $ \partial\Gamma_{\mathrm{bd}} = \overline{AC}\cup\overline{EB} $ and $ \partial\Gamma_{\mathrm{in}} = \overline{CD}\cup\overline{DE} $, and 
\begin{equation}\label{eq: psistar}
\psi_\star(\rv(s,t)) := f_0(t) \exp\left(-i\alpha_0 s - \tx\frac{1}{2}ist\right),
\end{equation}
where $(\alpha_0, f_0)$ is a minimizing pair realizing $E^{\mathrm{1D}}_0(\ell)$ and $ (s,t) $ are boundary coordinates of $ \Gamma_\beta(L,\ell) $, such that $ t \in [0, \ell] $ is the normal distance to the outer boundary and $ s \in [-L,L] $ the tangential coordinate (with the vertex in $ s = 0 $).

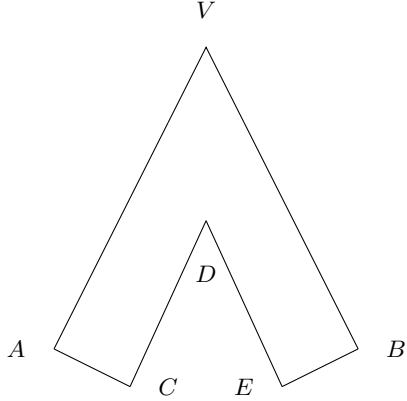
\begin{figure}[!ht]
	\begin{center}
	\begin{tikzpicture}
		\draw (0,0) -- (1,2);
		\draw (1,2) -- (1.5,3);
		\draw (1.5,3) -- (2,4) -- (2.5,3);
		\draw (2.5,3) -- (3,2);
		\draw (3,2) -- (4,0);
		\draw (0,0) -- (1,-0.5);
		\draw (4,0) -- (3,-0.5);
		\draw (1,-0.5) -- (2,1.7);
		\draw (3,-0.5) -- (2,1.7);
		\node at (2,4.5) {{\footnotesize $V$}};
		\node at (-0.5, 0) {{\footnotesize $A$}};
		\node at (4.5, 0) {{\footnotesize $B$}};
		\node at (1.5, -0.5) {{\footnotesize $C$}};
		\node at (2.5, -0.5) {{\footnotesize $E$}};
		\node at (2,1) {{\footnotesize $D$}};
	\end{tikzpicture}
	\caption{The region $ \corner $, where $ \beta $ is the opening angle $ \widehat{AVB} $, $ L = |\overline{AV}| = |\overline{VB}| $ and $ \ell = |\overline{AC}| = |\overline{EB}| $.}\label{fig: 1}
	\end{center}{}
\end{figure}

The heuristic motivation behind the definition \eqref{eq: Ecorner} of the effective energy is that, after blow-up around the singularity, in order to extract the precise contribution due to the presence of the corner, one has first to subtract the surface energy produced by the smooth part of the boundary, which equals exactly the first term on the r..h.s. of \eqref{eq: Ecorner} (see, e.g., \cite[Thm. 1.1]{CG}). Note also that, in the minimization in \eqref{eq: Ecorner}, the magnetic potential is fixed and the energy is minimized only over the order parameter. This is related to the very well known fact that, in the superconductivity layer surviving at the boundary, to a very good approximation, the magnetic field equals the applied one, so that, with a clever choice of the gauge, $ \aav $ can be replaced with $ \fv $. We omit further details for the sake of brevity (see \cite{CG2}).

An important consequence of the energy estimate of \cref{thm: CG1} is that superconductivity is robust along the boundary $ \partial \Omega $ with the possible exception of the corner points, i.e., $ \glm $ is non-vanishing there. More precisely, under the same assumptions of \cref{thm: CG1}, we proved \cite[Prop. 2.1]{CG2} that
\begin{equation}\label{eq: pan smooth}
		\lf\||\psi^{\mathrm{GL}}(\rv)| - f_\star(0) \ri\|_{L^\infty(\partial\Omega_{\mathrm{smooth}})} =o(1),
	\end{equation}
where $ \partial\Omega_{\mathrm{smooth}}:= \left\{\rv\in \partial\Omega\, \vert\, \dist(\rv, \Sigma)\geq c\eps|\log\eps|\right\} $.

\subsection{Main results}

As explained in details in \cite[\textsection 2.2]{CG2}, the explicit dependence of $ \ecorn $ on the angle $ \beta $ is not accessible and the definition of the effective energy is rather implicit. However, based on some geometric considerations, we formulated in \cite{CG2} the following conjecture.

\begin{conj}[GL corner energy]
	\label{conj}
	\mbox{}		\\
	For any $1<b<\Theta_0^{-1}$ and $\beta\in (0,2\pi)$, 
	\begin{equation}\label{eq: conjecture}
		E_{\mathrm{corner},\beta} = - (\pi-\beta)E_{\mathrm{corr}}
	\end{equation}
\end{conj}
The goal of this paper is to address \eqref{eq: conjecture} and show that, if the corner angle $ \beta $ is close to $ \pi $, \eqref{eq: conjecture} is in fact correct, at least to leading order. This is the content of the next theorem.

\begin{teo}[GL corner energy for almost flat angles]\label{teo: asympt E corner near pi}
	\mbox{}		\\
	Let $ 0 \leq \delta \ll 1 $ and
	$
		1<b<\Theta_0^{-1}
	$.
	Then, as $\ell\rightarrow +\infty$ and $L\rightarrow +\infty$,
	\beq\label{eq: Ecorner-delta}
 		E_{\mathrm{corner},\pi \pm\delta}(L,\ell)= \pm \delta E_{\mathrm{corr}} + \OO(\delta^{4/3} |\log\delta|) + \mathcal{O}(L^2\ell^{-\infty}).
	\eeq
\end{teo}

The above result is an asymptotic confirmation of \cref{conj} for almost flat angles. Unfortunately, we are not able to show that the remainder is identically zero, as it should be if one expects the conjecture to be true. There is however an important piece of information to extract from \eqref{eq: Ecorner-delta}: the fact that the correction is non-zero for both acute and obtuse angles marks the difference with the linear behavior close to $ \Hccc $. Indeed, in that case, it is expected and numerically confirmed that the effective energy is a monotone function of the angle but stays constant for angles larger or equal to $ \pi $. Here, on the opposite, we observe a symmetric behavior for acute and obtuse angles, at least for small variations around $\pi $: since one expects that $ E_{\mathrm{corr}} > 0 $ for $ 1 < b < \theo^{-1} $, this would imply that corners with obtuse angles increase the energy, while acute angles lower it.

The above result applies to the model problem \eqref{eq: def Ecorner}, but it has a direct consequence on the more general setting of a sample with a curvilinear polygonal boundary, provided the corners' angles are almost flat.

\begin{cor}[GL energy for domains with almost flat angles]
	\label{cor: GL almost flat}
	\mbox{}	\\
Let $ 0 \leq \delta \ll 1 $ and let $\Om\subset\mathbb{R}^2$ be a bounded domain with a finite number of corners along $\partial\Om$ with almost flat angles $\beta_j$, $j\in\Sigma$, i.e., such that $|\beta_j -\pi| \leq \delta \ll 1$. Then, for any fixed $ 1<b<\Theta_0^{-1} $, as $\eps\rightarrow 0$,
\beq\label{eq: asympt delta}
		E^{\mathrm{GL}}_\eps = \frac{|\partial\Om| E^{\mathrm{1D}}_0}{\eps} - 2\pi E_{\mathrm{corr}} + \OO(\delta^{4/3} |\log\delta|) + o_{\eps}(1).
\eeq
\end{cor}

The paper is organized as follows. We first discuss in detail the mathematical setting and formulate some preliminary results in \cref{sec: preliminaries}. The proofs of the main results above are then spelled in \cref{sec: proofs}. In \cref{sec: appendix}, we collect some useful results about the effective one-dimensional models involved in the proofs.

\section{Mathematical Setting and Preliminaries}
\label{sec: preliminaries}

Before facing the proofs of our main results, we provide more details about the mathematical setting we are going to consider. We study a wedge like domain $ \Gamma_{\pi \pm \delta}(L, \ell) $ described in \cref{fig: 2}. We recall that we denote by $\partial\Gamma_{\mathrm{out}}$ the outer portion of the boundary $\overline{AVB}$, while the inner part $ \overline{CDE} $ and the tangential components $ \overline{CV} \cup \overline{EB} $ are denoted by $ \bdi $ and $ \bdbd $, respectively. More precisely, using polar coordinates with origin at the vertex $V$ and axis along one of the outer sides of $ \Gamma_{\pi \pm \delta}(L, \ell) $, we get:
\begin{equation}
	\partial\Gamma_{\mathrm{out}} :=\{\rv\in\mathbb{R}^2\, \vert\, \rv = (\varrho,0),\, 0\leq \varrho\leq L\}\cup\{\rv\in\mathbb{R}^2\, \vert\, \rv  = (\varrho, \pi\pm\delta),\, 0\leq \varrho\leq L\},
\end{equation}	
\begin{equation}
\Gamma\equiv \Gamma_{\pi\pm\delta}(L,\ell) :=\{\rv\in\mathbb{R}^2\, \vert\, \dist(\rv,\partial\Gamma_{\mathrm{out}}\leq \ell)\},
\end{equation}
\begin{equation}
\partial\Gamma\equiv \partial\Gamma_{\pi\pm\delta}(L,\ell) = \partial\Gamma_{\mathrm{out}}\cup \partial\Gamma_{\mathrm{in}}\cup\partial\Gamma_{\mathrm{bd}},
\end{equation}
\begin{equation}
	\partial\Gamma_{\mathrm{in}}:=\{\rv\in\partial\Gamma\, \vert\, \dist(\rv, \partial\Gamma_{\mathrm{out}} = \ell)\}, \quad \partial\Gamma_{\mathrm{bd}} := \partial\Gamma \setminus (\partial\Gamma_{\mathrm{out}}\cup\partial\Gamma_{\mathrm{in}}).
\end{equation}
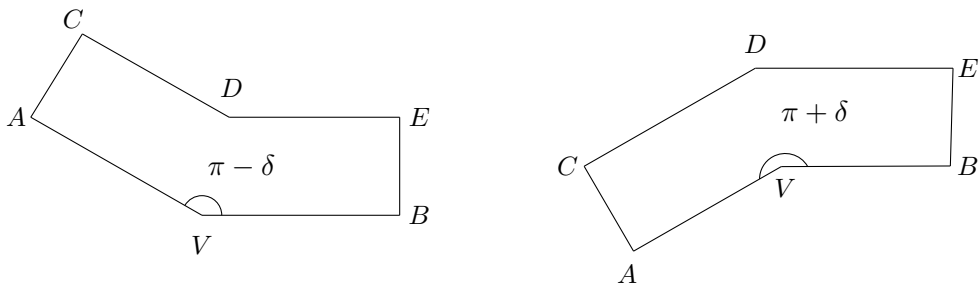
\begin{figure}[ht!]
	\begin{center}
	\begin{tikzpicture}[scale=1.3]
	\draw (2,0) -- (4,0);
	\draw (2,0) -- (0.268,1);
	\draw (0.268,1) -- (0.79, 1.85);
	\draw (0.79, 1.85) -- (2.275,1);
	\node at (2.4,0.5) {$\pi-\delta$};
	\draw (4,0) -- (4,1);
	\draw (4,1) -- (2.275,1);
	\draw[rotate around={210:(5,0)}] (2,0) -- (4,0);
	\draw[rotate around={210:(5,0)}] (2,0) -- (0.268,1);
	\draw[rotate around={210:(5,0)}] (0.268,1) -- (0.79, 1.85);
	\draw[rotate around={210:(5,0)}] (0.79, 1.85) -- (2.275,1);
	\draw[rotate around={210:(5,0)}] (4,0) -- (4,1);
	\draw[rotate around={210:(5,0)}](4,1) -- (2.275,1);
	\node at (8.2,1.05) {$\pi+\delta$};
	\node at (0.12,1) {\small{$A$}};
	\node at (0.7, 2) {\small{$C$}};
	\node at (2.3,1.3) {\small{$D$}};
	\node at (4.2,1) {\small{$E$}};
	\node at (4.2,0) {\small{$B$}};
	\node at (2,-0.3) {\small{$V$}};
	\node at (5.7,0.5) {\small{$C$}};
	\node at (6.3, -0.6) {\small{$A$}};
	\node at (7.9,0.28) {\small{$V$}};
	\node at (9.75,0.5) {\small{$B$}};
	\node at (9.75,1.5) {\small{$E$}};
	\node at (7.6,1.75) {\small{$D$}};
	\draw([shift=(30:1cm)]7.26,0) arc (30:180:0.26cm);
	\draw (2.2,0) arc (0:150:0.2);
	\end{tikzpicture}
	\caption{The angular regions $\Gamma_{\pi\pm\delta}(L,\ell)$.}
	\label{fig: 2}
	\end{center}
	\end{figure}
	
	The decay stated in \cref{teo: agmon} holds true in the region $ \Gamma $ as well, with the obvious adaptations. For later purposes however we formulate a more specific result applying to a subregion of longitudinal length of order $ 1 $.
	
	\begin{lem}
		\label{lem: agmon 2}
		\mbox{}\\
		Let $ S^{\pm}  \subset \Gamma^{\pm} $ be domains of the form 
		\beq
			S^{\pm} : = \lf\{ (s_{\pm},t_{\pm}) \in [-L,L] \times [0, \ell] \: \big| \: \bar{s}(t) \leq s_{\pm} \leq \bar{s}(t) + C \ri\}
		\eeq
		for some smooth increasing function $ \bar{s}(t) $. Then, for any $ b > 1 $, there exists a constant $ c(b) > 0  $, such that
		\beq
			\label{eq: agmon 2}
			\int_{S^{\pm}} \diff\textbf{r}\; e^{c(b) \: \dist (\textbf{r},\bdo)} \lf\{\lf| \psi_{\Gamma} \ri|^2+ \lf|\lf(\nabla+i  \fv \ri) \psi_{\Gamma} \ri|^2\ri\}  = \OO(1).
		\eeq
	\end{lem}
	
	\begin{proof}
		The proof is  a simple adaptation of the argument used to prove \cite[Lemma B.4]{CG2}.
	\end{proof}
	
	An alternative form of the Agmon estimate is provided by next Lemma, whose proof is given in \cite[Lemma B.5]{CG2}.
	
	\begin{lem}
		\label{lem: agmon 3}
		\mbox{}\\
		For any $ b >  1 $ there exists a finite constant $ C  $, such that
		\beq
			\label{eq: agmon 3}
			\lf| \psi_{\Gamma}(\rv) \ri| \leq C e^{- \frac{1}{2} c(b) \dist(\rv, \bdo)},
		\eeq
		where $ c(b) $ is the constant appearing in \eqref{eq: agmon 2}.
	\end{lem}

\subsection{Systems of coordinates}
Besides the aforementioned polar coordinates, we are going to use other sets of tubular-like coordinates we introduce here. This however calls for a split of the corner region into two subregions, in which we are allowed to use the tangential length along the boundary and the distance from the outer boundary as global coordinates. For concreteness, we consider only the case of opening angle $ \pi - \delta $ and provide a graphic representation of the domains $ \gpm $ in next \cref{fig: 3} (the adaptation to the case of angle $ \pi + \delta $ is trivial). More precisely, we set
\begin{equation}
	\partial\Gamma_{\mathrm{out}}^{+} := \lf\{ \rv\in \Gamma\, \vert\, \vartheta = 0, 0\leq \varrho\leq L \ri\} , \qquad
	\partial\Gamma_{\mathrm{out}}^{-} := \lf\{\rv\in \Gamma \, \vert\, \vartheta = \pi -\delta), 0\leq \varrho\leq L \ri\}, 
\end{equation}
\begin{equation}
	\Gamma^\pm := \lf\{ \rv\in \mathbb{R}^2\, \vert\, \dist(\rv, \partial\Gamma_{\mathrm{out}}^\pm) \leq \ell, 0 \leq \mp \lf( \vartheta - \vartheta_{\mathrm{bis}} \ri) \leq \vartheta_{\mathrm{bis}} \ri\},
\end{equation}
where $ \vartheta_{\mathrm{bis}} : = (\pi - \delta)/2 $ for short.
As before, we also set $ \partial\Gamma_{\mathrm{in}}^\pm := \partial \Gamma^{\pm} \cap \bdi $ and $ \partial\Gamma_{\mathrm{bd}}^\pm := \partial \Gamma^{\pm} \cap \bdbd $.

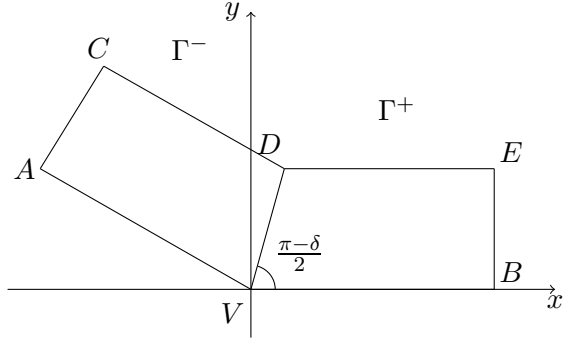
\begin{figure}[ht!]
	\begin{center}
	\begin{tikzpicture}[scale=1.6]
	\draw[->](0,0) -- (4.5,0);
	\draw[->](2,-0.4) -- (2, 2.3);
	\draw (2,0) -- (4,0);
	\draw (2,0) -- (2.275,1);
	\draw (2,0) -- (0.268,1);
	\draw (0.268,1) -- (0.79, 1.85);
	\draw (0.79, 1.85) -- (2.275,1);
	\draw (4,0) -- (4,1);
	\draw (4,1) -- (2.275,1);
	\node at (0.14,1) {$A$};
	\node at (0.75, 2) {$C$};
	\node at (1.5, 2) {$ \Gamma^-$};
	\node at (2.15,1.2) {$D$};
	\node at (4.14,1.14) {$E$};
	\node at (3.2,1.5) {$\Gamma^+$};
	\node at (4.14,0.14) {$B$};
	\node at (1.85,-0.2) {$V$};
	\node at (4.5, -0.1) {$x$};
	\node at (1.85, 2.3) {$y$};
	\draw (2.2,0) arc (0:73:0.2);
	\node at (2.4, 0.3) {$\frac{\pi-\delta}{2}$};
	\end{tikzpicture}
	\caption{The two subregions $\Gamma^\pm$.}
	\label{fig: 3}
	\end{center}
	\end{figure}

Let then $(x,y)$ be the cartesian coordinates centered at the vertex $V$, as in \cref{fig: 3}. The tubular coordinates in the regions $ \Gamma^{\pm} $ respectively read
\beq
	\label{eq: - coordinates}
	\begin{cases}
		s_+ := x,		\\
		t_+:=\dist\lf(\rv,\partial\Gamma_{\mathrm{out}}^+\ri) = y,\,
	\end{cases}
	\qquad		\forall \rv\in\Gamma^+,
\eeq
\beq
	\label{eq: + coordinates}
	\begin{cases}
			s_- = x \cos \delta - y \sin \delta, 	\\
			t_-:=\dist\lf(\rv,\partial\Gamma_{\mathrm{out}}^+\ri) = x\sin\delta+y\cos\delta,
	\end{cases}
	\qquad	\forall \rv\in\Gamma^-.
\eeq 
Hence, in both regions $ s_{\pm} $ measures the tangential length along the outer boundary, while $ t_{\pm} $ is the normal distance to it. Inside each region $ \Gamma^{\pm} $, the coordinates $ (s_{\pm}, t_{\pm}) $ identify a smooth diffeomorphism mapping $ \Gamma^{\pm} $ to the sets  $ \lf\{ 0\leq t_+\leq \ell,\;\tan(\delta/2)t_+\leq s_+\leq L \ri\} $ and $ \lf\{  0\leq t_-\leq \ell,\;-L\leq s_-\leq -t_-\tan(\delta/2) \ri\} $, which with a little abuse of notation we still denote by $ \Gamma^+ $ and $ \Gamma^- $, respectively. Note that the vertex of the corner is identified by the tubular coordinates $ s_{\pm} = 0, t_{\pm} = 0 $. 

\begin{remark}[Continuity of the normal coordinate]
	\label{rem: continuity coordinates}
	\mbox{}		\\
The coordinates defined above provide two patches to cover the corner region and are obviously discontinuous along the bisectrix. However, the discontinuity appears only in the tangential components $ s_{\pm} $. The values of $t_\pm$ indeed coincide on $ \partial \Gamma^+ \cap \partial \Gamma^- $, i.e., on the line $y=\tan\lf(\frac{\pi-\delta}{2}\ri)x$, where
\[
	t_+= \tan\lf( \frac{\pi-\delta}{2} \ri)x = \frac{x}{\tan(\delta/2)}, \qquad t_-=\bigg[\sin\delta+\frac{\cos\delta}{\tan(\delta/2)}\bigg]x=\tan\bigg(\frac{\pi-\delta}{2}\bigg)x.
\]
\end{remark}

We anticipate that the jump discontinuity of the tangential coordinates $ s_{\pm} $ is going to play a key role in the proof of our result and, in order to take it properly into account, we identify a transition region around the bisectrix (see \cref{fig: the four regions}), obtained by restricting the polar angle to the symmetric interval $ [\vartheta_<, \vartheta_>] $, with
\begin{equation}
	\vartheta_<:= \tx\frac{1}{2} (\pi-\delta-\gamma), \quad \vartheta_>:= \tx\frac{1}{2} (\pi-\delta+\gamma)
\end{equation}
with $ \gamma > 0 $ a (small) parameter to be chosen later.
We denote by $ \ggampm $ the corresponding subsets of $ \gpm $, i.e., explicitly, $ \ggampm:= \gpm \cap \lf\{ \lf|\vartheta - \vartheta_{\mathrm{bis}} \ri| \leq \gamma/2 \ri\} $. The complementary sets are denoted by $ \gcompm : = \Gamma^{\pm} \setminus \ggampm $, respectively.

\begin{figure}[ht!]
\begin{center}
\begin{tikzpicture}[scale=1.7]
\draw (2,0) -- (4,0);
\draw (2,0) -- (0.268,1);
\draw (0.268,1) -- (0.79, 1.85);
\draw (0.79, 1.85) -- (2.275,1);
\draw (0,0) -- (2,0);
\draw (4,0) -- (4,1);
\draw (4,1) -- (2.275,1);
\draw (2,0) -- (2.595,2.222);
\draw (2,0) -- (1.4,1.96);
\draw (2,0) -- (3.6, 1.6);
\node at (1.3,0.2) {$\delta$};
\node at (2.7,2.45) {$\vartheta_{\mathrm{bis}} = \frac{\pi-\delta}{2}$};
\node at (3.85,1.85) {$\vartheta_< = \frac{\pi-\delta-\gamma}{2}$};
\node at (1.6,2.2) {$\vartheta_> = \frac{\pi-\delta+\gamma}{2}$};
\node at (3.4,0.5) {$ \gcomp^+$};
\node at (2.5, 0.8) {$\ggam^+$};
\node at (1,1.3) {$\gcomp^-$};
\node at (2, 0.8) {$\ggam^-$};
\end{tikzpicture}
\caption{The four regions $ \ggampm, \gcompm $.}
\label{fig: the four regions}
\end{center}
\end{figure}

For later convenience, we observe that the relation between tubular and polar coordinates in $ \gpm $ is given by
\beq\label{eq: polar coord gamma - +}
	\begin{cases}
		s_+=\varrho\cos\vartheta,
		\\
		t_+=\varrho\sin\vartheta,
	\end{cases}
	\qquad
	\begin{cases}
	s_-=  \varrho\cos(\vartheta + \delta),
	\\
	t_-=\varrho\sin(\vartheta + \delta),
	\end{cases}
\eeq
respectively.

\subsection{Sketch of the proof}

Our main result is obtained by combining suitable upper and lower bounds  to the energy $ E_{\mathrm{corner},\pi \pm\delta}(L,\ell)$ working in different ways in all the subregions we introduced above. The heuristics behind the splitting of $ \corner $ just introduced is indeed that the leading contribution to the energy is given by the regions $ \gcompm $, i.e., $2LE^{1D}_0(\ell)$, while the regions $\ggampm $ are instead responsible for the correction associated with the jump of the tangential coordinate, i.e., $-\delta E_{\mathrm{corr}}$. As discussed in detail in \cite[\textsection 3.1]{CG2} the dominant term in the energy of a finite strip in the surface superconductivity regime is given by $ E^{1D}_0(\ell) $ times the length of the outer boundary ($ 2L $ in our case), where $ \ell $ is the width of the strip.
 
The presence of a non-trivial correction to such an energy is entirely due to the presence of the corner. Let us describe the underlying heuristics. Because of the singularity in the tangential coordinates $ s_{\pm} $, it is indeed impossible to glue together the optimizers of the effective models in the two components $ \gpm $ of the corner region. More precisely, to leading order, the boundary behavior of the order parameter is captured by a wave function of the factorized form (up to a global phase factor)
\bdm
	f_0(t) e^{- i \al_0 s - \frac{i}{2} s t},
\edm
where $ s$ and $ t $ stand for the tangential coordinate and the normal distance to the outer boundary, respectively, and $ \lf( f_0, \al_0 \ri) \in H^1(\R^+) \times \R $ is a minimizing pair for the variational problem \eqref{eq: def E1Dell} (see also \cref{sec: appendix}). Here, the term $ - \frac{1}{2} s t $ is a gauge phase needed to map the magnetic potential $ \mathbf{F} $ to the vector $ (-t, 0) $ in tubular coordinates.

Now, the modulus of the above ansatz is continuous through the bisectrix (see \cref{rem: continuity coordinates}), but, on the opposite, the phase has a jump discontinuity, inherited from the behavior of the tangential coordinates $ s_{\pm} $. In order to glue together the optimal profiles in $ \Gamma^{\pm} $ in the upper bound and then recover the leading order term of the energy, a non-trivial variation of the phase is needed, i.e., the phase of the minimizer must be a genuine two-dimensional function, unlike the model ansatz where the dependence on $ (s,t) $ is factorized. Concretely, we accommodate the transition by introducing a two-dimensional phase in $ \ggampm $ gluing together the phases $ - \al_0 s_{\pm} - \frac{1}{2} s_{\pm} t_{\pm} $ in $ \gpm $: with a little abuse of notation, we define $ \Phi: \Gamma \to \R $ as
\beq
	\Phi:=\begin{cases}
	\Phi_\pm(s_\pm,t_{\pm}), &\mbox{in }\, \gcompm,\\
	\Xi(\varrho,\vartheta), &\mbox{in } \ggampm,\\
	\end{cases}
\eeq
where the two phases $\Phi_\pm$ are defined as follows
\beq
		\Phi_\pm(s_\pm,t_\pm):=-\alpha_0 s_\pm - \tx\frac{1}{2}s_\pm t_\pm,\quad \mbox{in}\, \gcompm,
\eeq
and the function gluing them together is given in polar coordinates by
\beq\label{eq: Delta Phi}
		\Xi(\varrho,\vartheta) := \left[\alpha_0\varrho\sin\lf(\frac{\delta+\gamma}{2}\ri) +\frac{\varrho^2}{4}\sin(\delta + \gamma)\right]\lf(\frac{2\vartheta -\vartheta_<-\vartheta_>}{\gamma}\ri).
\eeq
Note that the function $ \Xi $ has been defined in such a way that
\begin{equation}
 	\lf. \Xi \ri|_{\vartheta = \vartheta_{<}} = \lf. \Phi_+ \ri|_{\vartheta = \vartheta_{<}},	\qquad		 \lf. \Xi \ri|_{\vartheta = \vartheta_{>}} = \lf. \Phi_- \ri|_{\vartheta = \vartheta_{>}},
\end{equation}
which ensures that the phase is continuous (in particular $ H^1(\Gamma) $). However, by estimating the difference $ \lf. \Xi \ri|_{\vartheta = \vartheta_{>}} - \lf. \Xi \ri|_{\vartheta = \vartheta_{<}} $, one immediately sees that the phase undergoes a jump of order $\OO(\delta+\gamma)$  in the region $ \ggam^+ \cup \ggam^- $, which is going to be very relevant in the proof. Note however that the choice of $ \Xi $ is rather arbitrary and any other function with the same properties would work.

On the other hand, the lower bound part of the proof does not involve the phase $ \Xi $, but an apparent discontinuity along the bisectrix emerges nevertheless: the key tools in the argument are a splitting technique to extract the energy to recover, by defining a pair of unknown functions $ u_{\pm} $ via
\bdm
	\psi_{\Gamma} = : f_0(t_{\pm}) e^{- i \al_0 s_{\pm} - \frac{1}{2} s_{\pm} t_{\pm}} u_{\pm},	\qquad	\mbox{in } \Gamma^{\pm}, 
\edm
and the consequent reduction of the problem to the minimization of a weighted functional $ \E_0 $ of $ u_{\pm} $. Note that, according to the upper bound heuristics, we expect $ u_{\pm} $ to be suitably close to $ 1 $, which would make the reduced energy $ \E_0 $ vanish identically. The key step in the lower bound is in fact the proof of the positivity of such an energy: this is done via a suitable integration by parts of the only non-positive term of $ \E_0 $ and exploiting the pointwise positivity of a one-dimensional {\it cost function} (see also \cref{sec: appendix}). There are however non-vanishing boundary terms along the bisectrix generated in the integration by parts, which provide an energy contribution necessary to recover the correction $ - \ecorr \delta $ and which are entirely due to the discontinuity of the phase $ - i \al_0 s_{\pm} - \frac{1}{2} s_{\pm} t_{\pm} $ (and, in particular, to the jump of the coordinate $ s_{\pm} $) on the bisectrix.

\section{Proofs}
\label{sec: proofs}

We denote for short by $\mathcal{G} $ the energy functional in \eqref{eq: Ecorner}, i.e., 
\begin{equation}
	\mathcal{G}[\psi; \Omega]:=\mathcal{E}^{\mathrm{GL}}_1[\psi, \fv;\Omega],
\end{equation}
and by $E_\Gamma$ and $ \psi_\Gamma $ its ground state energy and any corresponding minimizer for $ \Omega = \Gamma $, respectively (recall that $ \Gamma : = \corner $):
\begin{equation}\label{eq: def EGamma}
	E_\Gamma :=\inf_{\psi\in \dom(\Gamma)}\mathcal{G}[\psi; \Gamma] = \mathcal{G}[\psi_{\Gamma}; \Gamma],
\end{equation}
where $ \dom(\Gamma) : = \dom_\star(\Gamma_\beta(L,\ell)) $ is defined in \eqref{eq: dom star}.

As anticipated, we combine suitable upper and lower bounds to the energy to obtain the main result. For the sake of simplicity, we are going to spell in full detail only the proof for opening angle $ \pi - \delta $, while the adaptation to the angle $ \pi + \delta $ is discussed in \cref{sec: completion}.

\subsection{Energy upper bound}
We first prove an upper bound for the energy $E_\Gamma$. 

\begin{pro}[Energy upper bound]
	\label{pro: up bd}
	\mbox{}	\\
	For any fixed $ 1<b<\Theta_0^{-1} $, as $\ell, L \rightarrow +\infty$ and $\delta\rightarrow 0$,
	\beq
		\label{eq: ub}
 		E_\Gamma \leq 2 L E^{\mathrm{1D}}_0(\ell) - \delta E_{\mathrm{corr}} + \OO\big(\delta^{4/3}\big) + \exl.
	\eeq
\end{pro}

\begin{proof}
To prove the upper bound, we evaluate the energy of the trial state $ \trial \in \dom(\Gamma) $ given by (again with a little abuse of notation)
\beq\label{eq: psi trial}
		\psi_{\mathrm{trial}}:=
		\begin{cases}
			f_0(t_\pm)e^{i\Phi_\pm(s_\pm,t_\pm)}, & \mbox{in }  \gcompm, \\
			f_0(t_\pm)e^{i \Xi(\varrho, \vartheta)}, &\mbox{in }  \ggampm.
		\end{cases}
\eeq
Note that the trial state above satisfies the boundary conditions in $ \dom(\Gamma) $ and, moreover, $ \trial $ may be extended outside $ \Gamma $ expanding the domain of tubular coordinates.

We start by working in the region $ \gcomp^+$: denoting by $R^+ $ the rectangle obtained by completing $\gcomp^+$, i.e.,
\[
R^+:= \lf\{(s_+, t_+)\in \lf[0,L \ri]\times \lf[0,\ell \ri] \ri\},
\]
we obviously get
\beq
	\label{eq: gcomp}
	\mathcal{G} \lf[f_0 e^{i\Phi_+}; \gcomp^+ \ri] = \mathcal{G} \lf[f_0 e^{i\Phi_+}; R^+ \ri] - \mathcal{G} \lf[f_0 e^{i\Phi_+}; R^+\setminus \gcomp^+ \ri].
\eeq
Using that the boundary coordinates $ (s_+ , t_+) $ coincide with the cartesian ones in $ R $ (recall \eqref{eq: - coordinates}) and the identity
\bml{
	\label{eq: gauge change}
	\mathbf{F}(s_+,t_+) + \nabla_{(s_+,t_+)} \Phi_+ (s_+,t_+) =  \lf(- \tx\frac{1}{2} t_+ - \alpha_0 - \tx\frac{1}{2} t_+ \ri) \hat{\ev}_s + \lf( \tx\frac{1}{2} s_+  - \tx\frac{1}{2}  s_+ \ri) \hat{\ev}_t \\
	= - \lf( t_+ + \alpha_0  \ri) \hat{\ev}_s,
}
where we have denoted by $ \hat{\ev}_s $ and $ \hat{\ev}_t $ the unit vector in the tangential and normal directions, respectively, an easy calculation shows that (see, e.g., \cite[\textsection 3.1]{CG2})
\beq
	\label{eq: R^+}
	\mathcal{G} \lf[f_0 e^{i\Phi_+}; R^+ \ri] = L E^{\mathrm{1D}}_0(\ell),
\eeq
where we recall that $E^{\mathrm{1D}}_0(\ell)$ is the ground state energy of the 1D effective model in the interval $[0, \ell]$ introduced in \eqref{eq: def E1Dell}. On the other hand, we compute in $ R\setminus \gcomp^- $
\bml{
	\label{eq: R minus Gamma}
	\mathcal{G} \lf[f_0 e^{i\Phi_+}; R^+\setminus \gcomp^+ \ri] = \int_0^{\ell} \diff t \int_{0}^{t \tan \frac{\delta +\gamma}{2}} \diff s \lf\{|f_0'|^2 + (t+\alpha_0)^2 f_0^2 -\frac{1}{2b}(2f_0^2-f_0^4)\ri\}	\\
	= \frac{\delta +\gamma}{2}  \int_0^{\ell} \diff t \: t  \lf\{|f_0'|^2 + (t+\alpha_0)^2 f_0^2 -\frac{1}{2b}(2f_0^2-f_0^4)\ri\} + \OO(\gamma^3) ,
}
where we chose $ \gamma $ in such a way that
\beq
	\label{eq: gamma condition}
	\delta = \OO(\gamma).
\eeq 
Note also that the integral on the r.h.s. can be easily seen to be $ \OO(1) $ by exploiting the decay properties of $ f_0 $ (see \eqref{eq: f0 decay}) and the uniform boundedness of $ \alpha_0 $ (see \cref{sec: appendix}). 

Let us now consider the region $ \gcomp^- $: using the inverse transformation of \eqref{eq: + coordinates} together with the cartesian coordinate representation of the unit vectors $ \hat{\ev}_s = (\cos \delta, - \sin \delta) $, $ \hat{\ev}_t = (\sin \delta, \cos \delta) $, we compute
\begin{align}	
	&\fv(s_-,t_-)\cdot \hat{\ev}_s = \tx\frac{1}{2} \lf( s_-\sin\delta -t_-\cos\delta, s_-\cos\delta + t_-\sin\delta \ri) \cdot (\cos\delta, -\sin\delta) = -\tx\frac{1}{2} t_-,	\label{eq: gauge change 1 -}
	\\
	&\fv(s_-,t_-)\cdot \hat{\ev}_t = \tx\frac{1}{2} s_-,	\label{eq: gauge change 2 -}
\end{align}
so that
\bml{
	\lf|\lf(\nabla + i \fv \ri)\lf(f_0 e^{i\Phi_-}\ri)\ri|^{2}= f_0^2(t_-) \lf|\partial_s \Phi_-  - \tx\frac{1}{2} t_- \ri|^2 + \lf|f^{\prime}_0(t_-) + i \lf( \tx\frac{1}{2} s + \partial_{t_-} \Phi_- \ri) f_0(t_-) \ri|^2	\\
	= \lf| f^{\prime}_0(t_-) \ri|^2 + \lf( t_- + \alpha_0 \ri)^2 f_0^2(t_-).
}
Hence, we can proceed as for $ \gcomp^+ $ (see \eqref{eq: gcomp}, \eqref{eq: R^+} and \eqref{eq: R minus Gamma}), to get
\bdm
	\mathcal{G} \lf[f_0 e^{i\Phi_-}; \gcomp^- \ri] = L E_0^{\mathrm{1D}}(\ell) - \frac{\delta + \gamma}{2} \int_0^{\ell} \diff t \: t  \lf\{|f_0'|^2 + (t+\alpha_0)^2 f_0^2 -\frac{1}{2b}(2f_0^2-f_0^4)\ri\} + \OO(\gamma^3),
\edm
and combining this with the analogous result for $ \gcomp^+ $, we conclude that
\bml{
	\label{eq: energy far from bisectrix}
	\mathcal{G} \lf[\trial; \gcomp^- \cup \gcomp^+ \ri] = 2 L E_0^{\mathrm{1D}}(\ell) - (\delta +\gamma) \int_0^{\ell} \diff t \: t  \lf\{|f_0'|^2 + (t+\alpha_0)^2 f_0^2 -\frac{1}{2b}(2f_0^2-f_0^4)\ri\} \\
	+ \OO(\gamma^3). 
}

It remains to compute the energy contribution in the two regions $\Gamma^{\pm}_{\gamma}$ close to the bisectrix. There it is more convenient to work in polar coordinates. Since
\begin{align}\begin{split}\label{eq: a trial component rho theta}
	&\fv\cdot\hat{\ev}_{\varrho}:= \tx\frac{1}{2} \varrho \lf(-\sin\vartheta,\cos\vartheta \ri) \cdot \lf(\cos\vartheta, \sin\vartheta \ri) = 0,
\\
	&\fv\cdot\hat{\ev}_{\vartheta}:= \tx\frac{1}{2} \varrho \lf(-\sin\vartheta,\cos\vartheta \ri) \cdot \lf(-\sin\vartheta, \cos\vartheta \ri) = \tx\frac{1}{2} \varrho,
\end{split}\end{align}
we obtain
\[
		\mathcal{G} \lf[\psi_{\mathrm{trial}}; \ggam^+\ri] = \int_{\Gamma^{+}_{\gamma}}\diff\varrho\diff\vartheta \: \varrho\,\lf\{\lf| \partial_\varrho \trial \ri|^2 + \frac{1}{\varrho^2} \lf| \partial_{\vartheta} \trial + \tx\frac{i}{2} \varrho \trial \ri|^2-\frac{1}{2b}(2f_0^2-f_0^4)\ri\},
\]
where we have omitted for short the dependence of $f_0$ and $ \trial $ on $\varrho$ and $\vartheta$. Furthermore, we have
\bdm
	\lf| \partial_\varrho \trial \ri|^2 = \lf|\partial_\varrho \lf(f_0 e^{i\Xi} \ri) \ri|^2 = \lf|\partial_\varrho f_0(\varrho \sin \vartheta) \ri|^2 +  f_0^2 \lf|  \partial_\varrho \Xi(\varrho, \vartheta) \ri|^2,
\edm
and 
\beq
	\label{eq: deriv rho phase}
	\partial_\varrho \Xi = \left[\alpha_0\sin\bigg(\frac{\delta+\gamma}{2}\bigg) + \frac{\varrho}{2}\sin(\delta+\gamma)\right]\lf(\frac{2\vartheta-\vartheta_<-\vartheta_>}{\gamma}\ri) =  \OO(\gamma) + \OO(\varrho \gamma),
\eeq
under the assumption $ \delta = \OO(\gamma) $, so that (again by the decay \eqref{eq: f0 decay} of $ f_0 $)
\beq
	\int_{\Gamma^{+}_{\gamma}}\diff\varrho\diff\vartheta \: \varrho\: \lf| \partial_\varrho \trial \ri|^2 = \int_{\Gamma^{+}_{\gamma}}\diff\varrho\diff\vartheta \: \varrho \sin^2 \vartheta \lf| f_0^{\prime} (\varrho \sin\vartheta) \ri|^2 + \OO(\gamma^2).
\eeq
We now estimate the angular component of the kinetic energy in $ \gcomp^+ $:
\[
	\frac{1}{\varrho^2} \lf| \partial_{\vartheta} \trial + \tx\frac{i}{2} \varrho \trial \ri|^2
	= \cos^2 \vartheta \lf| f^{\prime}_0(\varrho \sin \vartheta) \ri|^2 +  f_0^2(\varrho \sin\vartheta) \lf| \frac{1}{\varrho} \partial_\vartheta \Xi + \frac{1}{2} \varrho \ri|^2.
\]
Since
\begin{equation}\label{eq: deriv theta phase}
	\frac{1}{\varrho} \partial_\vartheta \Xi
	 = \frac{2}{\gamma} \left[\alpha_0 \sin\lf(\frac{\delta+\gamma}{2}\ri) +\frac{\varrho}{4}\sin(\delta + \gamma)\right] = \frac{\varrho}{2} + \alpha_0 + \frac{\delta}{\gamma}\lf( \frac{\varrho}{2}+\alpha_0 \ri) +\OO(\gamma^2 ) +\OO(\varrho\gamma^2),
\end{equation}
we find 
\bml{
		\int_{\Gamma^{+}_{\gamma}}\diff\varrho\diff\vartheta \: \frac{1}{\varrho} \lf| \partial_{\vartheta} \trial + \tx\frac{i}{2} \varrho \trial \ri|^2
			= \int_{\Gamma^{+}_{\gamma}}\diff\varrho\diff\vartheta \: \lf\{ \varrho \cos^2 \vartheta \lf| f^{\prime}_0(\varrho \sin \vartheta) \ri|^2 \ri.	\\ 
			\lf. + \varrho f_0^2(\varrho \cos\vartheta) \lf[ \varrho+\alpha_0 + \tx\frac{\delta}{\gamma}\lf(\tx\frac{1}{2}\varrho+ \alpha_0 \ri)\ri]^2 \ri\} +\OO(\gamma^2).
}
We can now apply the estimates proven in \cref{lem: estimate f0}, to get
\bml{
	\mathcal{G} \lf[\psi_{\mathrm{trial}}; \ggam^+\ri] = \int_{\Gamma^{+}_{\gamma}}\diff\varrho\diff\vartheta \: \varrho \lf\{ \lf| f_0^{\prime} (\varrho)  \ri|^2 + f_0^2(\varrho) \lf[ \varrho + \alpha_0 + \tx\frac{\delta}{\gamma}\lf(\tx\frac{1}{2}\varrho + \alpha_0 \ri) \ri]^2 \ri.	\\
	\lf. - \tx\frac{1}{2b} \lf(2f_0^2(\varrho) -f_0^4(\varrho) \ri) \ri\} +  \OO(\gamma^2),
}
where we have also exploited the fact that $ \vartheta_> - \vartheta_< = \gamma $.
Now, we can approximate the integration domain as follows: let $ F(\varrho) $ be a function enjoying the same decay properties as $ f_0 $ or $ f_0^{\prime} $, then
\bml{
	\int_{\Gamma^{+}_{\gamma}}\diff\varrho\diff\vartheta \: \varrho \: F(\varrho) = \int_{\vartheta_<}^{\vartheta_{\mathrm{bis}}} \diff \vartheta \int_0^{\bar{\varrho}(\vartheta)} \diff \varrho \: \varrho \: F(\varrho)
	= \frac{\gamma}{2} \int_0^{\ell} \diff \varrho \: \varrho \: F(\varrho) + \int_{\vartheta_<}^{\vartheta_{\mathrm{bis}}} \diff \vartheta \int_{\ell}^{\bar{\varrho}(\vartheta)} \diff \varrho \: \varrho \: F(\varrho) \\
	=  \frac{\gamma}{2} \int_0^{\ell} \diff \varrho \: \varrho \: F(\varrho) + \OO(\gamma \ell^{-\infty}).
}
where we have denoted by $ (\bar{\varrho}(\vartheta), \vartheta) $, $ \vartheta \in [\vartheta_<, \vartheta_>] $, the polar coordinates of points belonging to the inner boundary of the region $ \ggam^+ \cup \ggam^- $ (note that $ \bar{\varrho}(\vartheta) \geq \ell $ for any $ \vartheta $). Hence,
\bml{
	\mathcal{G} \lf[\psi_{\mathrm{trial}}; \ggam^+\ri] = \frac{\gamma}{2} \int_{0}^{\ell} \diff\varrho  \: \varrho \lf\{ \lf| f_0^{\prime} (\varrho)  \ri|^2 + f_0^2(\varrho) \lf[ ( \varrho +\alpha_0)^2 + \frac{\delta}{\gamma}\lf( \varrho + \alpha_0 \ri) \lf(\varrho +  2\alpha_0 \ri) \ri] \ri.	\\
	\lf. - \frac{1}{2b} \lf(2f_0^2(\varrho) -f_0^4(\varrho) \ri) \ri\} +  \OO(\gamma^2) + \OO(\delta^2 \gamma^{-1}) + \OO(\gamma \ell^{-\infty}).
}

The estimate of the energy contribution from $ \ggam^- $ is perfectly analogous and we obtain
\bml{
	\mathcal{G} \lf[\psi_{\mathrm{trial}}; \ggam^+ \cup \ggam^- \ri] = \gamma \int_{0}^{\ell} \diff\varrho  \: \varrho \lf\{ \lf| f_0^{\prime} (\varrho)  \ri|^2 + ( \varrho +\alpha_0)^2 f_0^2(\varrho) - \tx\frac{1}{2b} \lf(2f_0^2(\varrho) -f_0^4(\varrho) \ri) \ri\} \\
	+ \delta \int_{0}^{\ell} \diff\varrho  \: \varrho \lf( \varrho + \alpha_0 \ri) \lf(\varrho +  2\alpha_0 \ri) f_0^2(\varrho)  +  \OO(\gamma^2) + \OO(\delta^2 \gamma^{-1}) + \OO(\gamma \ell^{-\infty}).
}
Combining this with \eqref{eq: energy far from bisectrix}, we finally get
\bml{
		\mathcal{G} \lf[\psi_{\mathrm{trial}}; \Gamma \ri] = 2 L E^{\mathrm{1D}}_0(\ell)
	 -\delta\int_0^{\ell} \diff t\, t\lf\{|f_0'|^2 +f_0^2\left[(t+\alpha_0)^2 - \alpha_0 \left(t+\alpha_0\right)\right]  - \tx\frac{1}{2b}(2f_0^2-f_0^4)\ri\} 
		\\
		+  \OO(\gamma^2) + \OO(\delta^2 \gamma^{-1}) + \OO(\gamma \ell^{-\infty}),
}
and, optimizing over $ \gamma $, i.e., choosing $ \gamma = \delta^{2/3} $, we conclude that
\beq
	E_{\Gamma} \leq 2 L E^{\mathrm{1D}}_0(\ell)  -\delta E_{\mathrm{corr}} + \OO(\delta^{4/3}) + \OO(\delta^{2/3} \ell^{-\infty} ).
\eeq
\end{proof}

\subsection{Energy Lower Bound}
We now prove the lower bound corresponding to the upper bound proved in Proposition \ref{pro: up bd}.

	\begin{pro}[Energy lower bound]\label{pro: lw bd}
		\mbox{}	\\
		For any fixed $ 1<b<\Theta_0^{-1} $, as $\ell,L \rightarrow +\infty$ and $\delta\rightarrow 0$,
		\begin{equation}\label{eq: lw bd}
 			E_{\Gamma}\geq 2 L E^{1D}_0(\ell) - \delta \ecorr + \OO(\delta^{4/3} |\log\delta|) + \mathcal{O}(L^2\ell^{-\infty}).
		\end{equation}
	\end{pro}

	The starting point of the lower bound proof is a suitable energy splitting, which is somewhat customary in the study of the GL and related functionals. 

	\begin{pro}[Energy splitting]\label{pro: energy splitting}
		\mbox{}	\\
		For any fixed $ 1<b<\Theta_0^{-1} $, let the function $ u : = u_+ \one_{\Gamma_+} + u_- \one_{\Gamma_-} \in C(\Gamma) $  be defined via
		\beq
			\label{eq: splitting u}
			\psi_\Gamma  =: u_{\pm} f_0 e^{i \Phi_\pm}, 	\qquad	\mbox{in } \Gamma^{\pm}.	
		\eeq
		Then, as $\ell, L \rightarrow +\infty$ and $ \delta \to 0 $,
		\beq\label{eq: splitting}
			E_{\Gamma} \geq 2L E^{1D}_0(\ell) - \delta \ecorr -\delta \int_0^{\ell} \diff t \: t \lf( t + \alO \ri) \lf( t + 2 \alO \ri) f_0^2 + \mathcal{E}_0[u] + \OO(\delta^3),
		\eeq 
		where we have used the compact notation $ \E_0[u] = \E_0[u^+; \Gamma^+] + \E_0[u^-; \Gamma^-]  $ and 
		\beq
			\mathcal{E}_0[u; \mathcal{D}]:=\displaystyle\int_{\mathcal{D}}\diff s \diff t \: f_0^2 \lf\{ \lf| \partial_t u \ri|^2 + \lf| \partial_s u \ri|^2 - 2 \lf( t + \alpha_0 \ri) j_s[u] + \tx\frac{1}{2b}f_0^2 \lf(1- |u |^2 \ri)^2 \ri\},
		\eeq
		with $ j_s[u] : = \hat{\ev}_s \cdot \jv[u] $ and the current $ \jv[u] $ defined as in \eqref{eq: sup current}.
	\end{pro}

\begin{proof}
	Let us consider first $ \Gamma^+ $: by \eqref{eq: splitting u}, dropping the label $ + $ on $ u $ for short,
	\[
		\mathcal{G}[\psi_\Gamma; \Gamma^+] =\int_{\Gamma^{+}}\diff s\diff t\, \lf\{\lf|\lf(\nabla + i \fv \ri) \lf(f_0\, u e^{i\Phi_+} \ri)\ri|^2 - \tx\frac{1}{2b}(2f_0^2|u|^2-f_0^4|u|^4) \ri\}.
	\]
	Using \eqref{eq: gauge change}, we get 
	\bmln{
		\lf|\lf(\nabla + i \fv \ri) \lf(f_0\, u\, e^{i\Phi_+} \ri)\ri|^2 = \lf|\lf(\nabla - i \lf(t + \alpha_0 \ri) \hat{\ev}_s \ri) \lf(f_0\, u \ri)\ri|^2 \\
		= |f_0'|^2 |u|^2 + f_0^2 |\partial_s u|^2 +  f_0^2 |\partial_t u|^2 + (t+\alpha_0)^2 f_0^2|u|^2								-2(t+\alpha_0)f_0^2\,j_s[u]+ f_0\partial_t f_0 \partial_t|u|^2.
	}
	Acting as in \cite{CR1}, we perform an integration by parts of the term $f_0(\partial_t f_0)\partial_t|u|^2$:
	\begin{equation}\label{eq: first integration gamma - delta}
		\int_{\Gamma^+}\diff s\diff t\, f_0(t) \nabla f_0 \cdot \nabla |u|^2 = \int_{\partial \Gamma^+}\diff\sigma\, f_0|u|^2 \hat{\nv}_+ \cdot \nabla f_0
	 -\int_{\Gamma^+}\diff s\diff t\, \lf\{ |f'_0|^2 |u|^2-f_0f''_0|u|^2 \ri\},
	\end{equation}
	where $\hat{\nv}_+$ is the outward normal unit vector along $\partial \Gamma^+$.
	The boundary of $ \Gamma^+ $ is composed of four segments but
	\beq
		\lf| \hat{\nv}_+ \cdot \nabla f_0 \ri| = \lf| f_0^{\prime} \ri| = 0,	\qquad		\mbox{on } \partial\Gamma^+ \cap \bdi  \mbox{ and on } \partial \Gamma^+ \cap \bdo
	\eeq
	while $ \lf| \hat{\nv}_+ \cdot \nabla f_0 \ri| = \lf| \hat{\ev}_s \cdot \nabla f_0 \ri| =0 $ on $ \partial \Gamma^+ \cap \bdbd $. Therefore, the boundary term above reduces to the one computed over the remaining portion of $ \partial \Gamma^+ $, i.e., the bisectrix of the region $ \Gamma $, that we denote by $ \partial \Gamma_{\mathrm{bis}} $.
	Hence, we obtain
	\bmln{
		\mathcal{G}[\psi_\Gamma; \Gamma^+] = \int_{\Gamma^+}\diff s\diff t\,f_0^2 \lf\{\frac{-f_0^{\prime\prime}}{f_0}|u|^2 + |\nabla u|^2  +(t+\alpha_0)^2 |u|^2
	-2(t+\alpha_0)\,j_s[u]+\frac{1}{2b}(2-f_0^2) \ri\}
	\\
	+\int_{\partial \Gamma_{\mathrm{bis}}} \diff\sigma\,  f_0 |u|^2\hat{\nv}_+ \cdot \nabla f_0.
	}
	which, via the variational equations for $f_0$, leads to
	\beq
		\label{eq: E_0 u1}
		\mathcal{G}[\psi_\Gamma; \Gamma^+] = - \frac{1}{2b} \int_{\Gamma^+}\diff s\diff t\,   f_0^4(t) + \E_0[u; \Gamma^+] +\int_{\partial \Gamma_{\mathrm{bis}}} \diff\sigma\,  f_0 |u|^2\hat{\nv}_+ \cdot \nabla f_0.
	\eeq
	
	Reproducing the computation in $ \Gamma^- $ and using \eqref{eq: gauge change 1 -} and \eqref{eq: gauge change 2 -} there, we end up with a similar expression, but, since $ \hat{\nv}_- = - \hat{\nv}_+ $, the two boundary terms cancel out, since  $ \lf| u_+ \ri| = \lf| u_- \ri| $ on $ \partial \Gamma_{\mathrm{bis}} $. The final result is then
	\beq
		\label{eq: E_0 u2}
		\mathcal{G}[\psi_\Gamma; \Gamma] = - \frac{1}{2b} \int_{\Gamma^+ \cup \Gamma^-}\diff s\diff t\,   f_0^4(t) + \E_0[u],
	\eeq
	and it just remains to compute the first term on the r.h.s.. Actually, such a computation has already been done in the upper bound proof, therefore we only sketch it:
	\bml{
		 - \frac{1}{2b} \int_{\Gamma^+ \cup \Gamma^-}\diff s\diff t\,   f_0^4(t) =  - \frac{1}{2b} \int_{R^+ \cup R^-}\diff s\diff t\,   f_0^4(t) + \frac{1}{2b} \int_{\lf( R^+ \setminus \Gamma^+ \ri) \cup \lf( R^- \setminus \Gamma^- \ri) }\diff s\diff t\,   f_0^4(t)	\\
		 = 2 L \eoneo(\ell) - \int_{\lf( R^+ \setminus \Gamma^+ \ri) \cup \lf( R^- \setminus \Gamma^- \ri) }\diff s\diff t\,  \lf\{ - f_0^{\prime\prime} + (t+\alO)^2 f_0 - \tx\frac{1}{2b} \lf(2 - f_0^2 \ri) f_0 \ri\} f_0	\\
		 = 2 L \eoneo - \delta \ecorr -\delta \int_0^{\ell} \diff t \: t \lf( t + \alO \ri) \lf( t + 2 \alO \ri) f_0^2 + \OO(\delta^3),
	}
	where we have computed
	\bmln{
		- \int_{R^+ \setminus \Gamma^+}\diff s\diff t\, f_0^{\prime\prime}(t) f_0(t) = - \int_0^{\ell \tan \delta/2} \diff s \int_{s/\tan \delta/2}^{\ell} \diff t \: f_0^{\prime\prime}(t) f_0(t) \\
		= \int_{R^+ \setminus \Gamma^+}\diff s\diff t\, {f_0^{\prime}(t)}^2 + \int_0^{\ell \tan \delta/2} \diff s \: f_0^{\prime}\lf(\tx\frac{s}{\tan \delta/2}\ri) f_0\lf(\tx\frac{s}{\tan \delta/2}\ri) \\ 		
		= \int_{R^+ \setminus \Gamma^+}\diff s\diff t\, {f_0^{\prime}(t)}^2 - \tx\frac{1}{2}  f_0^{2}\lf(0\ri) \tan \delta/2,
	}
	\bdm
		- \int_{R^- \setminus \Gamma^-}\diff s\diff t\, f_0^{\prime\prime}(t) f_0(t) = \int_{R^- \setminus \Gamma^-}\diff s\diff t\, {f_0^{\prime}(t)}^2+ \tx\frac{1}{2}  f_0^{2}\lf(0\ri) \tan \delta/2,
	\edm
	and estimated (see \eqref{eq: R minus Gamma})
	\bmln{
		\int_{R^{\pm} \setminus \Gamma^{\pm} }\diff s\diff t\,  \lf\{ {f_0^{\prime}}^2 + (t+\alO)^2 f_0^2 - \tx\frac{1}{2b} \lf(2 - f_0^2 \ri) f_0^2 \ri\}  \\ 
		= \delta \int_{0}^{\ell} \diff t \: t \lf\{ {f_0^{\prime}}^2 + (t+\alO)^2 f_0^2 - \tx\frac{1}{2b} \lf(2 - f_0^2 \ri) f_0^2 \ri\} + \OO(\delta^3).
	}
\end{proof}

After the splitting, we have almost isolated in \eqref{eq: splitting} the energy to extract from the lower bound. Indeed, the reduced energy $ \E_0[u] $ turns out to be positive on a rectangle, as a consequence of the pointwise positivity of the cost function $ \kol $ (see \eqref{eq: Kol}). There is however an additional unexpected term of order $ \delta $ (the third one on the r.h.s. of \eqref{eq: splitting}), which must cancel out in order to recover the correct energy. This is precisely the point where the phase discontinuity described in \cref{sec: preliminaries} comes into play: the region $ \Gamma $ is indeed obtained by gluing together two almost rectangular regions $\Gamma^{\pm} $ and $ \E_0[u]  $ is positive on it {\it up to} some boundary terms on the bisectrix, which exactly cancel the unwanted term in \eqref{eq: splitting}.

	\begin{pro}[Lower bound on the reduced energy]
		\label{pro: pos E0} 
		\mbox{}	\\
		Under the same assumptions of \cref{pro: energy splitting}, 
		\bml{\label{eq: pos E0}
			\mathcal{E}_0[u] \geq \delta \int_0^{\ell} \diff t \: t \lf( t + \alO \ri) \lf( t + 2 \alO \ri) f_0^2 + \frac{1}{4b} \int_{\Gamma} \diff s \diff t \:  f_0^4 \lf(1- |u|^2 \ri)^2	\\
		+ \OO(\delta^{4/3} |\log\delta|) + \OO(L^2 \ell^{-\infty}).
		}
	\end{pro}
	
\begin{proof}
	The idea of the proof goes back to the works on Bose-Einstein condensates (see, e.g., \cite{CPRY,CRv,CRY, R1}) and consists in an integration by parts of the current term in the energy, which is the only non-positive term there. This allows to express that contribution in terms of the {\it vorticity measure} $ \curl \jv[u] $, which in turn is pointwise bounded by the gradient square $ |\nabla u|^2 $. In conclusion, one gets a lower bound to the energy given by the integral of a suitable function times $  |\nabla u|^2 $. Such a function turns out to be the {\it cost function} $ \kol $ discussed in \cref{sec: appendix} (see \eqref{eq: Kol}) and its pointwise positivity is enough to complete the argument. As anticipated, however,  the shape of the domain  $ \Gamma = \Gamma^+ \cup \Gamma^-  $ is responsible for the emergence of some boundary terms along the bisectrix, which are actually crucial to recover the first term on the r.h.s. of \eqref{eq: pos E0}.
	
	Recalling the definition \eqref{eq: fol} of $ F_0 $, we write, 
\bml{
	\label{eq: div teo F0}
  	- 2\int_{\Gamma^{\pm}}\diff s\diff t\, (t+\alpha_0)f_0^2\, j_s[u_{\pm}]= - \int_{\Gamma^{\pm}}\diff s\diff t\, \nabla^{\perp} F_0 \cdot \jv[u_{\pm}]  = \int_{\Gamma^{\pm}}\diff s\diff t\, F_0 \nabla^\perp\cdot \mathbf{j}[u_{\pm}] \\
  - \int_{\partial \Gamma^{\pm}}\diff\sigma\, F_0\, \hat{\tav}_{\pm} \cdot \mathbf{j}[u_{\pm}],
}
where $\hat{\tav}_{\pm}$ are the unit tangential vectors to $\partial\Gamma^{\pm}$.
We first consider the first term on the r.h.s above: since $ \nabla^\perp \cdot \mathbf{j}[u]=- \partial_t \Im \lf( u^* \partial_s u \ri) + \partial_s \Im \lf( u^* \partial_t u \ri) = 2 \Im \lf( \partial_s u^* \partial_t u \ri) $ and $ F_0 \leq 0 $,
\begin{equation}\label{eq: int est F0}
	 \int_{\ggam^+}\diff s\diff t\, F_0 \nabla^\perp\cdot \mathbf{j}[u] \geq \int_{\ggam^+}\diff s\diff t\, F_0(t)  \lf| \nabla u \ri|^2,
\end{equation}
yielding
\bml{
	\label{eq: E0 lb 1}
	\E_0[u] \geq \int_{\Gamma^{\pm}} \diff s\diff t \: \kol(t)  \lf| \nabla u \ri|^2  + \frac{1}{2b} \int_{\Gamma} \diff s \diff t \:  f_0^4 \lf(1- |u |^2 \ri)^2 - \int_{\partial \Gamma^{\pm}}\diff\sigma\, F_0\, \hat{\tav}_{\pm} \cdot \mathbf{j}[u_{\pm}] \\
	\geq \frac{1}{2b} \int_{\Gamma} \diff s \diff t \:  f_0^4 \lf(1- |u |^2 \ri)^2 - \int_{\partial \Gamma^{\pm}}\diff\sigma\, F_0\, \hat{\tav}_{\pm} \cdot \mathbf{j}[u_{\pm}] + \OO(L^2 \ell^{-\infty}),
}
where we have used the pointwise positivity of $ \kol $ in $ I_{\bar{\ell}} $ (see \eqref{eq: kol positive}) and the exponential decay of $ \nabla \psi_{\Gamma} $, $ \psi_{\Gamma} $, $ f_0 $ and $ f_0^{\prime} $ provided by \eqref{eq: agmon 2}, \eqref{eq: agmon 3}, \eqref{eq: f0 decay} and \eqref{eq: f0 prime decay 2} to estimate the rest: by \eqref{eq: F0 bound},
\bmln{
	\int_{\Gamma^{\pm}} \diff s\diff t \: \kol(t)  \lf| \nabla u \ri|^2 \geq \int_{\Gamma^{\pm} \cap I_{\bar{\ell}}} \diff s\diff t \: \kol(t)  \lf| \nabla u \ri|^2 \geq - C \ell \int_{-L}^{L} \diff s \int_{\bar{\ell}}^{\ell} \diff t \: f_0^2(t) \lf| \nabla u \ri|^2	\\
	 \geq - C \ell \int_{-L}^{L} \diff s \int_{\bar{\ell}}^{\ell} \diff t \: \lf\{ \lf[ L + t + f_0^{-2} {f_0^{\prime}}^2   \ri] \lf| \psi_{\Gamma} \ri|^2 +  \lf| \nabla \psi_{\Gamma} \ri|^2 \ri\} = \OO(L^2 \ell^{-\infty}).
} 

Concerning the boundary term in \eqref{eq: E0 lb 1}, the vanishing of $ F_0 $ at $ t = 0 $ and $ t = \ell $ together with the vanishing of the $ t$-current of $ u $ at the boundary $ \bdbd $, thanks to the boundary conditions satisfied by $ \psi_{\Gamma} $, imply that the only boundary terms surviving are the ones along $ \partial \Gamma_{\mathrm{bis}} $. Exploiting the definitions of $ u_{\pm}  $ in \eqref{eq: splitting u}, which yield
\beqn
	\lf. \hat{\tav}_{+} \cdot \mathbf{j}[u_{+}] \ri|_{\partial \Gamma_{\mathrm{bis}}} & = & f_0^{-2} \lf[ j_{\varrho}\lf[ \psi_{\Gamma} \ri] - \lf| \psi_{\Gamma} \ri|^2 \lf( \alO \cos \thebis + \tx\frac{1}{2} \varrho \sin 2 \thebis \ri) \ri],		\\
	\lf. \hat{\tav}_{-} \cdot \mathbf{j}[u_{-}] \ri|_{\partial \Gamma_{\mathrm{bis}}} & = & - f_0^{-2} \lf[ j_{\varrho}\lf[ \psi_{\Gamma} \ri] - \lf| \psi_{\Gamma} \ri|^2 \lf( \alO \cos (\thebis+ \delta) + \tx\frac{1}{2} \varrho \sin (2 \thebis + 2\delta)  \ri) \ri],	
\eeqn
and the trivial identities $ \cos \thebis = \sin \delta/2 $, $ \cos(\thebis+\delta) = - \sin \delta/2 $, we obtain
\bml{
	 \lf. \hat{\tav}_{+} \cdot \mathbf{j}[u_{+}] \ri|_{\partial \Gamma_{\mathrm{bis}}} + \lf. \hat{\tav}_{-} \cdot \mathbf{j}[u_{-}] \ri|_{\partial \Gamma_{\mathrm{bis}}} = -\frac{2}{f_0^2(\varrho \cos \delta/2)} \lf| \psi_{\Gamma} \ri|^2 \lf( \alO \sin \delta/2 + \tx\frac{1}{2} \varrho \sin \delta \ri)	\\
	 = - \tx\frac{1}{2} \delta \lf(1 + \OO(\delta^2) \ri) \lf( \alO + \varrho \ri) \lf. \lf(  \lf| u_+ \ri|^2  + \lf| u_- \ri|^2  \ri)\ri|_{\partial \Gamma_{\mathrm{bis}}}.
}
Hence, the boundary terms in \eqref{eq: div teo F0} become (since $ \cos \delta/2 = 1 + \OO(\delta^2) $)
\bml{
	- \int_{\partial \Gamma^{+}}\diff\sigma\, F_0\, \hat{\tav}_{+} \cdot \mathbf{j}[u_{+}] - \int_{\partial \Gamma^{-}}\diff\sigma\, F_0\, \hat{\tav}_{-} \cdot \mathbf{j}[u_{-}] \\
	= -\frac{1}{2} \delta \lf(1 + \OO(\delta^2) \ri) \int_{0}^{\frac{\ell}{\cos\delta/2}} \diff \varrho \: F_0(\varrho \cos\delta/2)  \lf( \alO + \varrho \cos \delta/2 \ri) \lf. \lf(  \lf| u_+ \ri|^2  + \lf| u_- \ri|^2  \ri)\ri|_{\partial \Gamma_{\mathrm{bis}}}.
}
By \eqref{eq: F0 bound} and the decay of $ \psi_{\Gamma} $ in \eqref{eq: agmon 3}, we can get rid of the prefactor $ 1 + \OO(\delta^3) $ up to an error of order $ \delta^3 $. Furthermore, the smoothness of $ F_0 $, the bound \eqref{eq: F0 bound}, the decay of $ f_0 $ in \eqref{eq: f0 decay} and a direct integration by parts yield
\bml{
	-\int_{0}^{\frac{\ell}{\cos\delta/2}} \diff \varrho \: F_0(\varrho \cos\delta/2)  \lf( \alO + \varrho \cos \delta/2 \ri) = -\int_{0}^{\ell} \diff \varrho \: F_0(\varrho)  \lf( \alO + \varrho \ri) + \OO(\delta^2) \\
	= \int_{0}^{\ell} \diff \varrho \: F^{\prime}_0(\varrho)  \lf( \alO \varrho + \tx\frac{1}{2} \varrho^2 \ri) + \OO(\delta^2) = \int_{0}^{\ell} \diff t \: t (t + \alO) (t + 2 \alO) f_0^2(t) + \OO(\delta^2),
}
and plugging this into \eqref{eq: E0 lb 1}, we find
\bml{
	\label{eq: E0 lb 2}
	 \mathcal{E}_0[u] \geq \delta \int_0^{\ell} \diff t \: t \lf( t + \alO \ri) \lf( t + 2 \alO \ri) f_0^2 + \frac{1}{2b} \int_{\Gamma} \diff s \diff t \:  f_0^4 \lf(1- |u |^2 \ri)^2\\
	 + \frac{1}{2} \delta \int_{0}^{\frac{\ell}{\cos\delta/2}} \diff \varrho \: F_0(\varrho \cos\delta/2)  \lf( \alO + \varrho \cos \delta/2 \ri) \lf.\lf(  1 - \lf| u_+ \ri|^2  + 1 - \lf| u_- \ri|^2  \ri)\ri|_{\partial \Gamma_{\mathrm{bis}}} \\
	 + \OO(\delta^{3}) + \OO(L^2 \ell^{-\infty}).
	}
	
	Therefore, it just remains to bound from below the third term on the r.h.s. of \eqref{eq: E0 lb 2}. Let us consider only the contribution of $ u_+ $, since the other is perfectly analogous: let $ \chi $ be a smooth function such that $ \chi = 1 $ on $ \partial \Gamma_{\mathrm{bis}} $ and which vanishes on a line parallel to $ \partial \Gamma_{\mathrm{bis}} $ and passing through $ (\gamma, 0) $ for some $ \gamma > 0 $ to be chosen later. Concretely, we may take a function $ \chi(s) $ equal to $ 1 $ for $ s = t \sin \delta/2 $ and vanishing at $ s = t \sin \delta/2 + \gamma $. Moreover, we may assume that $ \lf| \chi^{\prime}(s) \ri| \leq C \gamma^{-1} $ and
	\bml{
		\label{eq: E0 lb 3}
		\delta \int_{0}^{\ell} \diff t \: F_0(t)  \lf( \alO + t \ri) \lf.\lf(  1 - \lf| u_+ \ri|^2 \ri) \chi  \ri|_{\partial \Gamma_{\mathrm{bis}}} = \delta \int_{0}^{\ell} \diff t \int_{t \sin \delta/2}^{\gamma + t \sin \delta/2} \diff s \: \chi^{\prime}(s) F_0(t)  \lf( \alO + t \ri) \lf(  1 - \lf| u_+ \ri|^2 \ri)
		\\
		- \delta \int_{0}^{\ell} \diff t \int_{t \sin \delta/2}^{\gamma + t \sin \delta/2} \diff s \: \chi(s) F_0(t)  \lf( \alO + t \ri) \partial_s \lf| u_+ \ri|^2.
	}
	
	We now estimate the two terms on the r.h.s. separately: acting as in the proof of \cref{lem: approx f0 and agmon est} and using once more \eqref{eq: F0 bound}, we estimate, for any $ a, d > 0 $,
	\bmln{
		\int_{0}^{\ell} \diff t \int_{t \sin \delta/2}^{\gamma + t \sin \delta/2} \diff s \: \chi^{\prime}(s) F_0(t)  \lf( \alO + t \ri) \lf(  1 - \lf| u_+ \ri|^2 \ri)	\\
		\geq \int_{0}^{\bar{\ell}} \diff t \int_{t \sin \delta/2}^{\gamma + t \sin \delta/2} \diff s \: \chi^{\prime}(s) F_0(t)  \lf( \alO + t \ri) \lf(  1 - \lf| u_+ \ri|^2 \ri)  + \exl \\
		\geq - \frac{C}{\gamma} \int_{0}^{\ell} \diff t \int_{t \sin \delta/2}^{\gamma + t \sin \delta/2} \diff s \: (1 + t) f_0^2 \lf( 1 - \lf| u_+ \ri|^2 \ri) + \exl  \\
		\geq - \frac{C}{\gamma} \lf[ \frac{1}{a} \int_{\Gamma^+} \diff s \diff t \:  f_0^4 \lf(1- |u_+ |^2 \ri)^2 + a d^3 \gamma + e^{- \frac{1}{2} d^2} \ri],
	}
	which leads to
	\bml{
		\label{eq: E0 lb 4}
		\delta \int_{0}^{\ell} \diff t \int_{t \sin \delta/2}^{\gamma + t \sin \delta/2} \diff s \: \chi^{\prime}(s) F_0(t)  \lf( \alO + t \ri) \lf(  1 - \lf| u_+ \ri|^2 \ri) \geq  - \frac{1}{4b} \int_{\Gamma^+} \diff s \diff t \:  f_0^4 \lf(1- |u_+ |^2 \ri)^2 \\
		+ \OO(\delta^2 \gamma^{-1} |\log\delta|^3),
	}
	after taking $ a = \frac{4b \delta}{C\gamma} $ and $ d = |\log\delta| $. For the other term, we directly apply \eqref{eq: agmon 2} to get
	\bml{
		\label{eq: E0 lb 5}
		- \delta \int_{0}^{\ell} \diff t \int_{t \sin \delta/2}^{\gamma + t \sin \delta/2} \diff s \: \chi(s) F_0(t)  \lf( \alO + t \ri) \partial_s \lf| u_+ \ri|^2 \\
		\geq - C \delta \int_{0}^{\bar{\ell}} \diff t \int_{t \sin \delta/2}^{\gamma + t \sin \delta/2} \diff s \: (1 + t) |\psi_{\Gamma}| \partial_s \lf| \psi_{\Gamma} \ri| + \exl = \OO(\delta \sqrt{\gamma}) + \exl,
	}
	via Cauchy-Schwarz inequality. Plugging \eqref{eq: E0 lb 4} and \eqref{eq: E0 lb 5} into \eqref{eq: E0 lb 3}, we finally get
	\bml{
		\delta \int_{0}^{\ell} \diff t \: F_0(t)  \lf( \alO + t \ri) \lf.\lf(  1 - \lf| u_+ \ri|^2 \ri) \chi  \ri|_{\partial \Gamma_{\mathrm{bis}}} + \frac{1}{4b} \int_{\Gamma^+} \diff s \diff t \:  f_0^4 \lf(1- |u_+ |^2 \ri)^2 \\
		= \OO(\delta^2 \gamma^{-1} |\log\delta|^3) + \OO(\delta \sqrt{\gamma}) + \exl = \OO(\delta^{4/3} |\log\delta|) + \exl,
	}
	after an optimization over $ \gamma $, i.e, taking $ \gamma = \delta^{2/3} |\log\delta|^2 $. Hence, we can use the second term on the r.h.s. of \eqref{eq: E0 lb 2} to compensate for the negative contribution on the l.h.s. above, to obtain the result.
\end{proof}

\subsection{Completion of the proofs}
\label{sec: completion}

We are now in position to complete the proofs of our mains results.

\begin{proof}{Proof of Theorem \ref{teo: asympt E corner near pi}}
	For corner angle $ \pi - \delta $, the result is obtained via a straightforward combination of \cref{pro: up bd}, \cref{pro: lw bd} and \eqref{eq: Ecorner}. Let us then comment on the adaptations needed for opening angle $ \pi + \delta $. The  strategy of the proof is indeed the same. In particular, in the upper bound, one has to split the region $ \Gamma $ into the subregions $ \gcompm $, $ \ggampm $ as before, with $\gamma \ll \delta$, and use the same trial state as in \eqref{eq: psi trial}. The outcome is the same estimate as in \eqref{eq: ub} with an opposite sign in front of $ \delta $, which is due to the fact that in this case $ R^{\pm} \subset \Gamma^{\pm} $. The lower bound, on the other hand, is proven in exactly the same way and, in particular, the phase singularity along $ \partial \Gamma_{\mathrm{bis}} $ appears as well.
\end{proof}

It just remains to address \cref{cor: GL almost flat}:

\begin{proof}[Proof of \cref{cor: GL almost flat}]
The result is a direct consequence of Gauss-Bonnet theorem for piecewise smooth domains combined with the energy asymptotics proved in \cref{thm: CG1} and \cref{teo: asympt E corner near pi}. In fact, by Gauss-Bonnet, one gets
\begin{equation}
	\int_{\partial\Omega_{\mathrm{smooth}}}\diff\mathbf{s}\, \mathfrak{K}(\mathrm{s}) + \sum_{j\in\Sigma}(\pi- \beta_j) = 2\pi
\end{equation}
If we use the asymptotics \eqref{eq: Ecorner-delta} for $E_{\mathrm{corner}, \beta_j}$, assuming that  $|\beta_j - \pi| \ll 1$, $ \forall j \in \Sigma $, together with the energy asymptotics \eqref{eq: asym CG2}, choosing $L,\ell=\mathcal{O}(|\log\eps|)$ as in \cite{CR3,CG2}, we get \eqref{eq: asympt delta}. 
\end{proof}

\appendix

\section{One-dimensional Effective Models}
\label{sec: appendix}

In this Appendix we collect some well known results about the effective one-dimensional problems, which are known to play a role in surface superconductivity. We refer to \cite{CR1,CR2,CR3,CDR,CG2} for more details (see, in particular, \cite[Appendix A]{CG2} for a review).

\subsection{Effective model on the half-line}
\label{sec: 1d disc}

The effective model useful to describe the behavior of the order parameter in the surface superconductivity regime is given in first approximation by the following energy functional:
\begin{equation}
	\label{eq: fone}
	\fone_{\star,\alpha}[f] := \displaystyle\int^{+\infty}_0 \mbox{dt} \left\{ |\partial_t f|^2 + (t+\alpha)^2 f^2 -\frac{1}{2b} (2f^2-f^4)\right\},
\end{equation}
where $ t $ is the rescaled distance from the outer boundary and $ \alpha \in \R  $.

For any fixed $ \alpha \in \R $, the functional \eqref{eq: fone} has a unique minimizer in the domain $ \onedom = \lf\{ f \in H^1(\R^{{+}}; \R) \: | \: t f(t) \in L^2(\R^{{+}}) \ri\} $, which is strictly positive and monotonically decreasing for $ t $ large enough. We set $	\eones : = \inf_{\alpha \in \R} \eone_{\star,\alpha} = \inf_{\alpha \in \R} \inf_{f \in \onedom} \fone_{\star,\alpha}[f] $ and denote by $(\alpha_\star, f_\star)\in \mathbb{R}\times \dom^{1\mathrm{D}}$ any minimizing pair for \eqref{eq: fone}. We recall some key properties of the minimization \eqref{eq: fone}:
\begin{itemize}
	\item variational equation for $ \fs $:
		\beq\label{eq: fs var}
			- \fs^{\prime\prime}+(t+\alpha_0)^2\fs = \tx\frac{1}{b}(1 - \fs^2)\fs,\quad 	\mathrm{with}\,\,\,\fs^{\prime}(0) = 0;
		\eeq
 	\item optimality of the phase $ \as $:
		\beq
			\label{eq: optimal as}
			\int_0^{+\infty} \diff t \: \lf( t + \as \ri) \fs^2(t) = 0;
		\eeq
	\item the ground state energy can be expressed as
		\beq
			\eones = - \frac{1}{2b} \int_0^{+\infty} \diff t \: \fs^4(t) < 0.
		\eeq
	\item $ \fs $ decays exponentially in the distance from the boundary \cite[Prop. 3.3]{CR1}:
		\beq
			\label{eq: fs decay}
			\fs \leq C  \exp \lf\{ - \tx\frac{1}{2} \lf( t + \as \ri)^2 \ri\}.
		\eeq
\end{itemize}

\subsection{Effective model on a finite interval}
	\label{sec: 1d no curv}
	We now discuss a variant of the 1D effective model above by minimizing the energy on a finite interval $ I_{\ell} : = [0,\ell] $, $ \ell \gg 1 $, rather than in the whole of $ \R^+ $, i.e., we set
	\beq
		\label{eq: foneo}
		\fone_{\alpha,\ell}[f] := \displaystyle\int^{\ell}_0 \mbox{dt} \left\{ |\partial_t f|^2 + (t+\alpha)^2 f^2 -\frac{1}{2b} (2f^2-f^4)\right\},	\qquad		\eoneo : = \inf_{\alpha \in \R} \inf_{f \in H^1(I_{\ell})} \fone_{\alpha,\ell}[f].
	\eeq
	Exactly as before, one can prove that there is a unique minimizing pair $ \lf(\alpha_0, f_0 \ri) \in \R \times H^1(I_{\ell}) $ of $ \fone_{\alpha,\ell} $. Moreover, $ f_0$ satisfies the same variational equation of $f_\star$ in \eqref{eq: fs var} in the interval $ [0,\ell] $ with $ \alpha_0 $ in place of $ \alpha_\star $, i.e.,
	\begin{equation}\label{eq: var eq f0}
		- f_0^{\prime\prime}+(t+\alpha_0)^2f_0 = \tx\frac{1}{b}(1 - f_0^2)f_0.
	\end{equation}
	In addition, $ f_0 $ satisfies Neumann boundary conditions
	\beq
		\label{eq: fol nbc}
		\fol^{\prime}(0) = \fol^{\prime}(\ell) = 0.
	\eeq
	Furthermore, all the properties \eqref{eq: optimal as} -- \eqref{eq: fs decay}  hold true for $ \fol $ in the interval $[0,\ell]$. In particular, $ f_0 $ is monotonically decreasing for $ t \geq t_0 $, where  $ t_0 $ is the unique maximum point of $ f_0 $ and it satisfies
	\beq
		\label{eq: t0}
		0 < t_0 \leq |\alpha_0| + \frac{1}{\sqrt{b}},
	\eeq
	and
	\beq
		\label{eq: f0 decay}
		f_0(t) \leq C  \exp \lf\{ - \tx\frac{1}{2} \lf( t + \alO \ri)^2 \ri\}.
	\eeq
	A similar estimate holds true for the derivative of $ f_0 $ (see \cite{CG2}[Lemma A.1]): for any $1<b<\Theta_0^{-1}$ and $\overline{t}\in [1,\ell]$, there exists a finite constant $C>0$, such that
	\beqn
	 	\label{eq: f0 prime decay 1}
		|f^{\prime}_0(t)| & \leq & C\exp^{-\frac{1}{4} t^2},	\qquad \forall\, t\in [0,\ell],	\\
	 	\label{eq: f0 prime decay 2}
		|f^\prime_0(t)| & \leq & C\, \overline{t}^3\, f_0(t),		\qquad \forall\, t\in[0,\overline{t}].
	 \eeqn
	
\subsection{Cost function}
\label{sec: cost function}	
A very important object constructed over $ f_0 $ is the {\it cost function}
	\beq	
	\label{eq: Kol}
		\kol(t) : =\lf(1 - d_\ell\ri) \fol^2(t) + \Fol(t),\qquad d_{\ell} = \OO(\ell^{-4}),
	\eeq
	where (by optimality of $ \alO $)
	\beq
		\label{eq: fol}
		\Fol(t) : = 2\int_0^t\diff\eta\, \lf(\eta + \al_0 \ri) \fol^2(\eta) = - 2 \int_t^\ell \diff \eta \, \lf(\eta + \al_0 \ri) \fol^2(\eta),
	\eeq
	The most important property of $ \kol $ is its pointwise positivity in the relevant region, i.e., where $ f_0 $ is not exponentially small in $\ell $: for any $ 1 {<} b < \theo^{-1} $,
	\beq
		\label{eq: kol positive}
		\kol(t) \geq 0, 	\qquad		\mbox{for any } t \in \annol,
	\eeq
	where $I_{\bar{\ell}}$ is defined as 
	\bdm\label{eq: def Iell}
		\annol : = \lf\{ t \in (0,\ell) \: \big| \: \fol(t) \geq \ell^{3} \fol(\ell) \ri\} = \lf[0, \bar{\ell} \ri].
	\edm
	Finally, we underline that \eqref{eq: fs var} and \eqref{eq: fs decay}, imply that $\overline{\ell} = \ell + \OO(1)$ and, as a consequence, 
	\begin{equation}
		f_0(t) = \mathcal{O}(\ell^{-\infty}),\, \qquad \mbox{for }\, t\in [0,\ell]\setminus I_{\bar{\ell}}. 
	\end{equation}
	Note that, as a by-product of the positivity of $ \kol $, we also get
	\beq
		\label{eq: F0 bound}
		\lf| F_0(t) \ri| \leq  
		\begin{cases}
		f_0^2(t),			&	\mbox{for } t \in I_{\bar{\ell}},		\\
		C \ell f_0^2(t),	&  	\mbox{for } t \in I^{\mathrm{c}}_{\bar{\ell}},
		\end{cases}
	\eeq
	since, inside $ I_{\bar{\ell}} $, $ |F_0(t)| \leq f_0^2(t) $, while, for $ t \geq \bar{\ell} $, we can use the monotonicity of $ f_0 $ to bound
	\bdm
		|F_0(t)| = \int_t^{\ell} \diff \eta \: 2( \eta + \alO) f_0^2(\eta) \leq C \ell f_0^2(t).
	\edm
	
\subsection{Useful estimates close to $ \partial \Gamma_{\mathrm{bis}} $}
\label{sec: estimates}

We prove here some useful bounds in the regions $ \ggampm $. The key observation is that, in those regions, $ \vartheta = \vartheta_{\mathrm{bis}} + \OO(\gamma) = \frac{\pi}{2} + \OO( \delta + \gamma) $, so that
\[
	\sin\vartheta = 1 + \OO(\gamma^2),\qquad \cos\vartheta = \OO(\gamma),
\]
under the assumption $ \delta = \OO(\gamma) $ (recall \eqref{eq: gamma condition}) and this leads to some other useful approximations.

\begin{lem}\label{lem: estimate f0}
\mbox{}\\
Let $ \delta = \OO(\gamma)$. Then, for any $\vartheta\in [\vartheta_<,\vartheta_>] $, there exist two finite constant $c, C > 0 $, such that
\beqn
	\lf|f_0(\varrho\sin\vartheta) - f_0(\varrho) \ri| & \leq & C \gamma e^{- c\varrho^2}, 
	\label{eq: approx f0 1}	\\
	\lf| f^{\prime}_0(\varrho\sin\vartheta) - f^\prime_0(\varrho) \ri| & \leq & C\gamma^2 e^{- c\varrho^2}. \label{eq: approx f0 2}	
\eeqn
\end{lem}

\begin{proof}
For the proof of \eqref{eq: approx f0 1} we simply observe that a straightforward application of Taylor formula yields
\[
	f_0(\varrho\sin\vartheta) = f_0(\varrho)  + f_0^{\prime}(\varrho\sin\tilde\vartheta)\lf(\vartheta - \tx\frac{\pi}{2}\ri),
\]
for some $\tilde\vartheta\in (\vartheta_<,\vartheta_>) $. We now use \eqref{eq: f0 prime decay 1} to bound
\[
	\lf|f_0^\prime(\varrho\sin\tilde\vartheta) \ri| \leq C e^{-\frac{1}{4}\varrho^2\sin^2\tilde\vartheta }\leq C e^{-c\varrho^2},
\]
for some $c>0$, which implies the result.

To prove \eqref{eq: approx f0 2}, we use again the Taylor formula applied to $ f_0^{\prime} $  and exploit the variational equation \eqref{eq: var eq f0} to control $ f_0^{\prime\prime} $, which yields the desired estimate. We omit the computation for the sake of brevity.
\end{proof}

\begin{lem}
	\label{lem: approx f0 and agmon est}
	\mbox{}\\
	Let $ \delta = \OO(\gamma)$, $ \beta \geq 0 $ and $ a > 1 $. Then, for any $\vartheta\in[\vartheta_<,\vartheta_>]$ and any  function $u \in C(\Gamma)$, we get
	\begin{multline}\label{eq: gamma f0}
		\int_{\ggampm}\diff\vartheta \diff\varrho\,\varrho^\beta \: f_0^2(\varrho \sin\vartheta) |u|^2 - \gamma\int_{0}^{\ell}\diff\varrho\, \varrho^\beta f_0^2(\varrho) \geq \frac{1}{a} \int_{\ggampm}\diff\vartheta \diff\varrho \, \varrho \: f_0^4(\varrho\sin\vartheta) \lf(1-|u|^2\ri)^2
		\\
		 + \OO(\gamma^{2}) + \OO\big( \gamma a \lf|\log a\ri|^{\beta} \big).
	\end{multline}
\end{lem}
	
\begin{proof}
	We write
	\beq \label{eq: pm f0}
		\int_{\ggampm}\diff\vartheta \diff\varrho\,\varrho^\beta \: f_0^2(\varrho \sin\vartheta) |u|^2= \int_{\ggampm}\diff\vartheta \diff\varrho\,\varrho^\beta  f_0^2(\varrho\sin\vartheta)
	+ \int_{\ggampm}\diff\vartheta \diff\varrho\,\varrho^\beta  f_0^2(\varrho\sin\vartheta) \lf(|u|^2-1 \ri)
	\eeq
	and exploit the approximation of $f_0(\varrho\sin\vartheta)$ proved in \cref{lem: estimate f0},  specifically \eqref{eq: approx f0 1}, to obtain
	\beq
		\label{eq: approx f0 logdelta}
		\int_{\ggampm}\diff\vartheta \diff\varrho\,\varrho^\beta  f_0^2(\varrho\sin\vartheta)
		= \gamma\int_{0}^{\ell}\diff\varrho\,\varrho^\beta f_0^2(\varrho) + \OO(\gamma^2),
	\eeq
	where we have replaced $\ell/\sin\vartheta$ and $\ell/\sin(\vartheta + \delta)$ with $\ell$ in the integration up to a small error $\exl$. 
	
	We now estimate the second term in \eqref{eq: pm f0}. Let us consider only $ \ggam^+ $ since the estimate in $ \ggam^- $ is identical: for any $ d \in (0,\ell) $
	\bml{ \label{eq: split logdelta}
			\int_{\ggam^+}\diff\vartheta \diff\varrho\,\varrho^\beta  f_0^2(\varrho\sin\vartheta) \lf(|u|^2-1 \ri)= \int_{\vartheta_<}^{\vartheta_{\mathrm{bis}}}\int_0^{d}\diff\vartheta \diff\varrho\,\varrho^\beta  f_0^2(\varrho\sin\vartheta) \lf(|u|^2-1 \ri)
			\\
			+\int_{\vartheta_<}^{\vartheta_{\mathrm{bis}}}\int_d^{\frac{\ell}{\sin \vartheta}}\diff\vartheta \diff\varrho\,\varrho^\beta  f_0^2(\varrho\sin\vartheta) \lf(|u|^2-1 \ri).
	}
	We estimate the first term on the r.h.s. of \eqref{eq: split logdelta} as
	\begin{multline*}
	\int_{\vartheta_<}^{\vartheta_{\mathrm{bis}}}\int_0^{d}\diff\vartheta \diff\varrho\,\varrho^\beta  f_0^2(\varrho\sin\vartheta) \lf(|u|^2-1 \ri)
	\\
	\geq - \frac{a \gamma}{2} \int_0^{d}\diff\varrho\, \varrho^{2\beta - 1} - \frac{1}{a}\int_{\ggam^+}\diff\vartheta \diff\varrho\,\varrho \:  f_0^4(\varrho \sin \vartheta) \lf(1 - |u|^2\ri)^2
	\\
	\geq  -\frac{1}{a} \int_{\ggam^+}\diff\vartheta \diff\varrho\,\varrho \:  f_0^4(\varrho \sin \vartheta) \lf(1 - |u|^2\ri)^2 - C a \gamma d^{2 \beta}.
	\end{multline*}
	
	On the other hand, the second term on the r.h.s. of \eqref{eq: split logdelta} decays exponentially in $d$: by the exponential decay of $ f_0 $ in \eqref{eq: f0 decay} and the fact that $ \sin \vartheta = 1 + \OO( \gamma^2) > 0 $, we deduce that 
	\bml{ \label{eq: estimate expoentially small delta}
		\int_{\vartheta_<}^{\vartheta_{\mathrm{bis}}}\int_d^{\frac{\ell}{\sin\vartheta}}\diff\vartheta \diff\varrho\,\varrho^\beta  f_0^2(\varrho\sin\vartheta) \lf(|u|^2-1 \ri) \geq - \int_{\vartheta_<}^{\vartheta_{\mathrm{bis}}}\int_d^{\frac{\ell}{\sin\vartheta}}\diff\vartheta \diff\varrho\,\varrho^\beta  f_0^2(\varrho\sin\vartheta) \\
		\geq - C \gamma  e^{- \frac{1}{2} d^2},
	} 
	which completes the proof after an optimization over $ d $.
	\end{proof}

\end{document}